\documentclass[a4paper,USenglish,cleveref,autoref,nameinlink,numberwithinsect]{lipics-v2021}
\title{On the Parameterized Complexity of Grundy Domination and Zero Forcing Problems} %TODO Please add

\author{Robert Scheffler}{Institute of Mathematics, Brandenburg University of Technology, Cottbus, Germany}{robert.scheffler@b-tu.de}{https://orcid.org/0000-0001-6007-4202}{}%TODO mandatory, please use full name; only 1 author per \author macro; first two parameters are mandatory, other parameters can be empty. Please provide at least the name of the affiliation and the country. The full address is optional. Use additional curly braces to indicate the correct name splitting when the last name consists of multiple name parts.

\authorrunning{R. Scheffler} %TODO mandatory. First: Use abbreviated first/middle names. Second (only in severe cases): Use first author plus 'et al.'

\Copyright{Robert Scheffler} %TODO mandatory, please use full first names. LIPIcs license is "CC-BY";  http://creativecommons.org/licenses/by/3.0/


 %TODO mandatory; please add comma-separated list of keywords

\acknowledgements{I would like to thank Ekki Köhler for bringing the Grundy domination problem to my attention. I~also thank Jesse Beisegel and Max Göttlicher for helpful discussions. Furthermore, I~thank the anonymous reviewers for helpful remarks, in particular for their hint to the problems \UD{} and \UI{} as well as their suggestion to include \cref{tab:force-w1} in the proof of \cref{thm:force-w1}.}%optional

\nolinenumbers %uncomment to disable line numbering

\usepackage{tikz}
\tikzstyle{vertex}=[draw, circle]
\tikzstyle{svertex}=[draw, circle, inner sep=1.5pt]
\tikzstyle{all}=[ultra thick]
\usetikzlibrary{arrows,fit,shapes.geometric,positioning}
\usepackage{mathdots}
\usepackage{booktabs}
\usepackage{todonotes}

\usepackage{bbding}
\usepackage{pifont}

\definecolor{myblue}{HTML}{007FFF}

\newcommand{\T}{\ensuremath{\mathfrak{T}}}
\newcommand{\fR}{\ensuremath{\mathfrak{R}}}
\newcommand{\cP}{\ensuremath{\mathcal{P}}}
\newcommand{\C}{\ensuremath{\mathcal{C}}}
\renewcommand{\S}{\ensuremath{\mathcal{S}}}
\newcommand{\D}{\ensuremath{\mathfrak{D}}}
\newcommand{\E}{\ensuremath{\mathcal{E}}}
\newcommand{\K}{\ensuremath{\mathcal{K}}}
\renewcommand{\H}{\ensuremath{\mathcal{H}}}
\newcommand{\R}{\ensuremath{\mathcal{R}}}
\newcommand{\A}{\ensuremath{\mathcal{A}}}
\renewcommand{\P}{\ensuremath{\mathsf{P}}}
\newcommand{\NP}{\ensuremath{\mathsf{NP}}}
\newcommand{\paraNP}{\ensuremath{\mathsf{para\text{-}NP}}}
\newcommand{\XP}{\ensuremath{\mathsf{XP}}}
\newcommand{\FPT}{\ensuremath{\mathsf{FPT}}}
\newcommand{\W}{\ensuremath{\mathsf{W}}}
\newcommand{\WT}{\ensuremath{\W[t]}}
\newcommand{\WOne}{\ensuremath{\W[1]}}
\newcommand{\WTwo}{\ensuremath{\W[2]}}

\renewcommand{\O}{\ensuremath{\mathcal{O}}}
\newcommand{\N}{\ensuremath{\mathbb{N}}}
\newcommand{\false}{\texttt{false}}
\newcommand{\true}{\texttt{true}}
\newcommand{\ZTD}{\ensuremath{\{Z,T,D\}}}
\newcommand{\ZT}{\ensuremath{\{Z,T\}}}
\newcommand{\TD}{\ensuremath{\{T,D\}}}
\newcommand{\ZD}{\ensuremath{\{Z,D\}}}

\newcommand{\crule}[1]{\ensuremath{\overset{\scriptscriptstyle #1}{\longrightarrow}}}
\newcommand{\z}{\crule{Z}}
\renewcommand{\t}{\crule{T}}
\renewcommand{\d}{\crule{D}}
\newcommand{\x}{\crule{X}}
\newcommand{\y}{\crule{Y}}

\newcommand{\BIP}{\textsc{One-Sided Grundy Total Domination}}
\newcommand{\GN}{\textsc{Grundy Domination}}
\newcommand{\TN}{\textsc{Grundy Total Domination}}
\newcommand{\ZN}{\textsc{Z-Grundy Domination}}
\newcommand{\LN}{\textsc{L-Grundy Domination}}
\newcommand{\LLN}{\textsc{Local L-Grundy Domination}}
\newcommand{\LLNC}{\textsc{(Local) L-Grundy Domination}}
\newcommand{\MCP}{\textsc{Multicolored Clique}}
\newcommand{\HYPER}{\textsc{Grundy Covering in Hypergraphs}}
\newcommand{\ZFS}{\textsc{Zero Forcing Set}}
\newcommand{\FORCE}[1]{\ensuremath{\textsc{Force}\,(#1)}}
\newcommand{\CONC}[1]{\textsc{(Connected) #1}}
\newcommand{\CON}[1]{\textsc{Connected #1}}
\newcommand{\TOT}[1]{\textsc{Total #1}}
\newcommand{\DUAL}[1]{\textsc{Dual #1}}
\newcommand{\PDS}{\textsc{Power Dominating Set}}
\newcommand{\UD}{\textsc{Upper Domination}}
\newcommand{\UI}{\textsc{Upper Irredundance}}

\newtheorem{assumption}{Assumption}

\hideLIPIcs

\begin{document}

\maketitle

\begin{abstract}
    We consider two different problem families that deal with domination in graphs. On the one hand, we focus on dominating sequences. In such a sequence, every vertex dominates some vertex of the graph that was not dominated by any earlier vertex in the sequence. The problem of finding the longest dominating sequence is known as \GN{}. Depending on whether the closed or the open neighborhoods are used for domination, there are three other versions of this problem: \TN{}, \LN{}, and \ZN{.} We show that all four problem variants are \WOne-complete when parameterized by the solution size. 

    On the other hand, we consider the family of zero forcing problems which form the parametric duals of the Grundy domination problems. In these problems, one looks for the smallest set of vertices initially colored blue such that certain color change rules are able to color all other vertices blue. It has been shown that the dual of \ZN{}, known as \ZFS, is in \FPT{} when parameterized by the treewidth or the solution size. We extend the treewidth result to the other three variants of zero forcing and their respective Grundy domination problems. Our algorithm also implies an \FPT{} algorithm for \GN{} when parameterized by the number of vertices that are not in the dominating sequence. In contrast, we show that \LN{} is \WOne-hard for that parameter.
\end{abstract}

\section{Introduction}
\label{intro}

\subparagraph{Grundy Domination Problems} A dominating set of a graph $G$ is a set $S$ of vertices such that all vertices of $G$ are either in $S$ or have a neighbor in $S$. When we want to find such a dominating set, we could greedily take vertices that dominate some vertex that was not dominated by any vertex chosen before. This greedy algorithm produces a so-called \emph{dominating sequence}\footnote{The names of the vertex sequences used in the following are taken from Lin~\cite{lin2019zero}. Note that they may differ from the names used in the papers that introduced the notions.} $(v_1, \dots, v_k)$ of vertices such that $N[v_i] \setminus \bigcup_{j=1}^{i-1} N[v_j]$ is not empty for every $i \in \{1, \dots, k\}$.  Obviously, this greedy algorithm is not guaranteed to find the smallest dominating set. Therefore, it is an interesting question how bad the greedy solution might perform, i.e., what is the longest dominating sequence in the graph. This graph parameter was introduced by Bre\v{s}ar et al.~\cite{bresar2014dominating} as the \emph{Grundy domination number}. 

In~\cite{bresar2016total}, a similar notion was introduced for the concept of total domination where every vertex has to have some neighbor in the set $S$. A \emph{total dominating sequence} $(v_1, \dots, v_k)$ fulfills the condition that $N(v_i) \setminus \bigcup_{j=1}^{i-1} N(v_j)$ is not empty for every $i \in \{1, \dots, k\}$. The maximal length of a total dominating sequence is called the \emph{Grundy total domination number}~\cite{bresar2016total}. 

Observe that for dominating sequences, one considers the closed neighborhood of the vertices while for total dominating sequences the open neighborhoods are considered. The authors of~\cite{bresar2017grundy} introduce two versions of dominating sequences that mix the two notions of neighborhood. In a \emph{$Z$-sequence}, $N(v_i) \setminus \bigcup_{j=1}^{i-1} N[v_j]$ is not empty, while in an \emph{$L$-sequence} $N[v_i] \setminus \bigcup_{j=1}^{i-1} N(v_j)$ is not empty for all $i \in \{1, \dots, k\}$. The \emph{$Z$-Grundy domination number} and \emph{$L$-Grundy domination number} are the maximum lengths of such sequences. 

So far, the algorithmic research on these four types of dominating sequence problems mostly focused on classical complexity theory. There are several \NP-completeness proofs for these problems on different graph classes~\cite{bresar2017grundy,bresar2020grundy,bresar2014dominating,bresar2016total,bresar2018total,bresar2022computational} (see \cref{fig:reductions} for an overview of the used reductions). Besides this, there have been given polynomial-time algorithms on graph classes such as forests~\cite{bresar2020grundy,bresar2014dominating,bresar2018total}, cographs~\cite{bresar2014dominating}, chain graphs~\cite{bresar2022computational,bresar2023computation}, interval graphs~\cite{bresar2016dominating} as well as bipartite distance-hereditary graphs~\cite{bresar2018total}. Note that the complexity of different problems may vary dramatically on the same class of graphs. For example, there is a simple linear-time algorithm that computes the Grundy domination number of split graphs~\cite{bresar2014dominating}, while computing the Grundy total domination number of split graphs is \NP-hard~\cite{bresar2018total}.

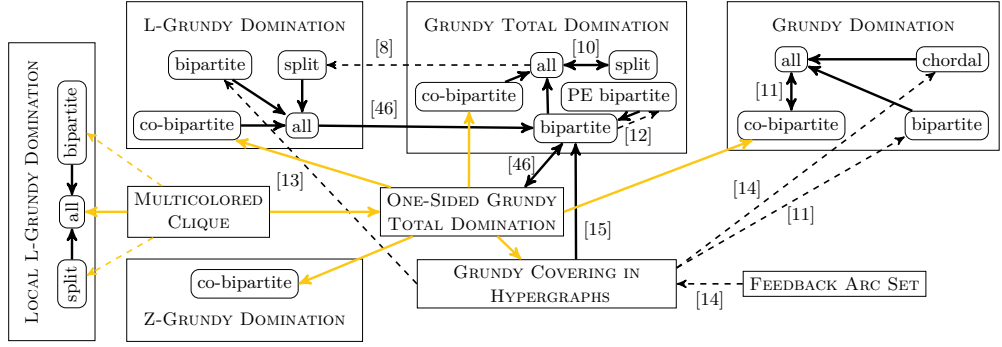
\begin{figure}
    \centering
    \resizebox{\textwidth}{!}{
    \tikzstyle{fpt}=[very thick, -stealth']
\tikzstyle{fptboth}=[very thick, stealth'-stealth']
\tikzstyle{fptmy}=[very thick, -stealth', lipicsYellow]
\tikzstyle{poly}=[dashed, thick, -stealth']
\tikzstyle{polymy}=[dashed, thick, -stealth',lipicsYellow]
\tikzstyle{problem}=[draw, rectangle]
\tikzstyle{class}=[draw, rectangle, rounded corners]

\begin{tikzpicture}
\footnotesize
    \node[problem, inner sep=5pt] (FAS) at (4.5,-7.5) {\textsc{Feedback Arc Set}};
    \node[problem, align=center, inner sep=5pt] (BIP) at (-1.8,-6.25) {\textsc{One-Sided Grundy} \\ \textsc{Total Domination}};
    \node[align=center, problem, minimum width=4.5cm, inner sep=5pt] (H) at (-0.5,-7.5) {\textsc{Grundy Covering in} \\ \textsc{ Hypergraphs}};

    \node[problem, align=center, inner sep=5pt] (MCP) at (-6.55,-6.25) {\textsc{Multicolored} \\ \textsc{Clique}};

    \node (GNtext) at (5,-3) {\textsc{Grundy Domination}};
    \node[class] (GNcb) at (3.75,-4.75) {co-bipartite};
    \node[class] (GNb) at (6.45,-4.75) {bipartite};
    \node[class] (GNc) at (6.55,-3.6) {chordal};
    \node[class] (GNa) at (3.75,-3.6) {all};
    \node (GN) [problem, fit={(GNtext) (GNa) (GNb) (GNc) (GNcb)}, inner sep=5pt] {};

    \node (TNtext) at (-0.6,-3) {\textsc{Grundy Total Domination}};
    \node[class] (TNa) at (-0.5,-3.7) {all};
    \node[class] (TNb) at (-0,-4.8) {bipartite};
    \node[class] (TNs) at (1,-3.7) {split};
    \node[class] (TNcb) at (-1.85,-4.25) {co-bipartite};
    \node[class, align=center] (TNpeb) at (0.75,-4.25) {PE bipartite};

    \node (TN) [problem, fit={(TNtext) (TNa) (TNb) (TNs) (TNpeb) (TNcb)}, inner sep=5pt] {};

    \node (LNtext) at (-6,-3) {\textsc{L-Grundy Domination}};
    \node[class] (LNb) at (-6.35,-3.7) {bipartite};
    \node[class] (LNcb) at (-6.75,-4.8) {co-bipartite};
    \node[class] (LNa) at (-4.75,-4.8) {all};
    \node[class] (LNs) at (-4.75,-3.7) {split};
    \node (LN) [problem, fit={(LNtext) (LNa) (LNb) (LNa) (LNcb) (LNs)}, inner sep=5pt] {};

    \node[rotate=90, align=right] (LLNtext) at (-9.45,-5.9) {\textsc{Local L-Grundy}  \textsc{Domination}};
    \node[class, rotate=90] (LLNa) at (-8.75,-6.25) {all};
    \node[class, rotate=90] (LLNb) at (-8.75,-4.7) {bipartite};
    \node[class, rotate=90] (LLNs) at (-8.75,-7.5) {split};
    \node (LLN) [problem, fit={(LLNtext) (LLNa) (LLNs) (LLNb)}, inner sep=5pt] {};

    \node (ZNtext) at (-5.75,-8.1) {\textsc{Z-Grundy Domination}};
    \node[class] (ZNcb) at (-5.75,-7.5) {co-bipartite};
    \node (ZN) [problem, fit={(ZNtext) (ZNcb)}, inner sep=5pt] {};

    \draw[fpt] (H.42) to node[pos=0.25,right]{\cite{bresar2016total}} (TNb);
    \draw[fpt] (LNa) to node[pos=0.305,above]{\cite{lin2019zero}} (TNb);
    \draw[fptboth] (TNa) to node[above]{\cite{bresar2018total}} (TNs);
    \draw[fpt] (TNb.152) to (TNa);
    \draw[fptboth] (GNa) to node[left]{\cite{bresar2023computation}} (GNcb);
    \draw[fpt] (LNb) to (LNa);
    \draw[fpt] (GNc) to (GNa);
    \draw[fpt] (GNb) to (GNa);
    \draw[fpt] (LNs) to (LNa);
    \draw[fpt] (LNcb) to (LNa);
    \draw[fpt] (TNcb) to (TNa);
    \draw[fpt] (TNpeb.325) to (TNb.11);

    \draw[fptboth] (BIP.25) to node[left]{\cite{lin2019zero}} (TNb);

    \draw[poly] (TNa) to node[pos=0.8,above right]{\cite{bresar2020grundy}} (LNs);
    \draw[poly] (FAS) to node[below]{\cite{bresar2014dominating}} (H);
    \draw[poly] (H.180) to node[pos=.5,left=0.15cm]{\cite{bresar2017grundy}} (LNb);
    \draw[poly] (H.7) to node[pos=0.4,right=0.2cm]{\cite{bresar2023computation}} (GNb.195);
    \draw[poly] (H.7) to node[pos=0.4,left=0.2cm]{\cite{bresar2014dominating}} (GNc);
    \draw[poly] (TNb.0) to node[pos=0.5, below]{\cite{bresar2022computational}} (TNpeb.340);

    \draw[polymy] (MCP) to (LLNb);
    \draw[polymy] (MCP) to (LLNs);

    \draw[fpt] (LLNs) to (LLNa);
    \draw[fpt] (LLNb) to (LLNa);

    \draw[fptmy] (BIP.97) to (TNcb);
    \draw[fptmy] (BIP) to (LNcb);
    \draw[fptmy] (BIP.0) to (GNcb);
    \draw[fptmy] (BIP) to (H);
    \draw[fptmy] (BIP) to (ZNcb.0);
    \draw[fptmy] (MCP) to (BIP);
    \draw[fptmy] (MCP) to (LLNa);
\end{tikzpicture}
    }
    \caption{The reductions for Grundy domination problems. Rectangular nodes represent problems, rounded corner nodes within these problem nodes represent restrictions of the problem to the respective graph class. Black arrows represent reductions given in the literature, yellow arrows stand for reductions given in this paper. Thick solid arrows represent polynomial-time \FPT{} reductions, while dashed arrows stand for polynomial-time reductions that are not \FPT{} reductions. Note that trivial reductions from a graph class to some superclass do not need references.}
    \label{fig:reductions}
\end{figure}

Two problems related to \GN{} are \UD{}~\cite{cheston1990computational} and \UI~\cite{cockayne1978properties}. Both problems ask for a largest set $S$ of vertices that is \emph{irredundant}, i.e., every vertex of $S$ has an element in its closed neighborhood that is not in the closed neighborhood of any other vertex of $S$. In \UD{}, the set $S$ must additionally form a dominating set of the graph. It is easy to see that the maximal size of an irredundant set is a lower bound on the Grundy domination number. In fact, a set is irredundant if and only if every permutation of its elements forms a dominating sequence. Nevertheless, the difference between these values can be arbitrarily large.\footnote{Consider two cliques $\{u_0, \dots, u_n\}$ and $\{v_0, \dots, v_n\}$ such that $u_i$ and $v_j$ are adjacent if and only if $i > j$. Then the largest irredundant set has size~2. However, there is a dominating sequence of length $n+2$. If we replace the edges between the cliques by a matching that saturates all vertices but $u_0$ and $v_0$, then the largest irredundant set and dominating sequence have size $n+2$ while the largest dominating irredundant set has size~2.} 

\subparagraph{Zero Forcing Problems} 

Several graph problems are based on the following situation. We have a graph where some vertices are colored blue and all other vertices are colored white. If there is a blue vertex that has exactly one white neighbor, then this white neighbor is allowed to be colored blue. The main goal of the graph problems is to find the smallest initial set of blue vertices that allows to color all vertices of the graph blue. This color change scheme has been mentioned in very different applications, for example quantum control~\cite{burgarth2007full} or the supervision of electrical power systems~\cite{brueni2005pmu}. The first study of this concept from a graph-theoretical point of view was given by Haynes et al.~\cite{haynes2002domination} who introduced the problem \PDS{}. In this problem, one looks for the smallest vertex set $S$ such that coloring the vertices of $S$ and their neighbors blue allows the described color change rule to color all vertices of the graph blue. An alternative problem formulation, now known as \ZFS{}, was introduced around the year 2008 independently by several researchers~\cite{aazami2008hardness,aim2008zero,severini2008nondiscriminatory}. Here, only the vertices of $S$ are colored blue at the beginning. The authors of~\cite{aim2008zero} considered the problem for their study of the minimum rank of a graph $G$, i.e. the smallest rank of a symmetric matrix that can be constructed by replacing the ones in the adjacency matrix of $G$ by arbitrary non-zero reals. Later, other color change rules have been studied for other matrix problems~\cite{barioli2013parameters,ima2010minimum}. Furthermore, \emph{connected forcing sets}~\cite{brimkov2016characterizations,brimkov2017complexity,davila2018bounds} and \emph{total forcing sets}~\cite{davila2019total,davila2015bounding} have been considered, where the set induces a connected graph or a graph without isolated vertices, respectively. A survey on all these concepts can be found in the textbook of Hogben et al.~\cite{hogben2022inverse}.

It is shown in~\cite{bresar2017grundy} that \ZFS{} is strongly related to \ZN{}. There is a zero forcing set of size $\leq k$ in an $n$-vertex graph if and only if there is a $Z$-sequence of length $\geq n - k$. This result was later extended by Lin~\cite{lin2019zero} to the other versions of dominating sequences and other variants of zero forcing sets. So from a parameterized complexity perspective, the zero forcing problems are the parametric duals of the Grundy domination problems just as \textsc{Independent Set} is the dual of \textsc{Vertex Cover} or \textsc{Nonblocker}~\cite{dehne2006nonblocker} is the dual of \textsc{Dominating Set}.

The first complexity result on \ZFS{} was an \NP-hardness proof for a weighted version given by Aazami~\cite{aazami2008hardness}. Using a different name for the problem,  Yang~\cite{yang2013fast} extended this to the unweighted case. The \NP-hardness of the connected and the total case has been proven by Brimkov and Hicks~\cite{brimkov2017complexity} as well as Davila and Henning~\cite{davila2019total}, respectively. Polynomial-time algorithms have been given for several graph classes including trees~\cite{severini2008nondiscriminatory} and proper interval graphs~\cite{cazals2019power}.

\subparagraph{Parameterized Results}
The number of results from the field of parameterized complexity on Grundy domination problems and zero forcing problems is relatively small. Gonzalez~\cite{gonzalez2024problemas} presented \XP{} algorithms for the Grundy (total) domination number when parameterized by the treewidth of the graph. Aazami~\cite{aazami2008hardness} gave an \FPT{} algorithm for \ZFS{} when parameterized by the treewidth. This also implies that the $Z$-Grundy domination number can be computed in \FPT{} time for that parameter. Besides this, Aazami showed that the minimum size of a zero forcing set is an upper bound on the treewidth. The combination of these two results leads to an \FPT{} algorithm for \ZFS{} when parameterized by the solution size, a fact that was not explicitly stated by Aazami. 
Recently, Bhyravarapu et al.~\cite{bhyravarapu2025parameterized} again gave an \FPT{} result on the treewidth and the solution size for \ZFS{}, which follows the same ideas as Aazami~\cite{aazami2008hardness}.

In contrast to \ZFS{}, the problem \PDS{} has been shown to be \WTwo-hard when parameterized by the solution size, and, thus, probably does not allow for an \FPT{} algorithm~\mbox{\cite{guo2008improved,kneis2006parameterized}}. This hardness result has recently been strengthened to $\W[P]$-completeness by Bl{\"a}sius and G{\"o}ttlicher~\cite{blaesius2025efficient}.

Besides this, Cazals et al.~\cite{cazals2019power} showed that a precolored version of \ZFS{} is \WTwo-hard. Here, some of the vertices in the forcing set are already fixed and the parameter is the maximal number of additional vertices that might be added to the forcing set. Most recently, G{\"o}ttlicher et al.~\cite{goettlicher2026from} improved this to $\W[P]$-completeness. Furthermore, the authors studied the parameterized complexity of problem variants where more than one vertex can be forced or the number of forcing rounds is bounded. The above mentioned problems \UD{} and \UI{} are also \WOne-hard, while their parametric duals are in \FPT{}~\cite{bazgan2018many,downey2000complexity}.

\subparagraph{Our Contribution}

We extensively study the parameterized complexity of Grundy domination problems and zero forcing problems (see \cref{tab:results} for an overview of the main results).
In \cref{sec:grundy}, we show that all four Grundy domination problems mentioned above are \WOne-complete when parameterized by the solution size (see \cref{fig:reductions} for the used reductions). Furthermore, we show the \WOne-completeness of a problem variant for hypergraphs as well as for the connected variants, where the vertices of the dominating sequence must induce a connected subgraph. Another variant of dominating sequences -- called \emph{local $L$-sequences} -- is introduced in \cref{sec:lln}, where we show that this problem as well as its connected variant are \WOne-complete as well.

Aside from these local $L$-sequences, \cref{sec:zero} focuses on zero forcing problems, the parametric duals of the Grundy domination problems. We show that the duals of all five Grundy domination problems can be solved in $2^{\O(k^2)} \cdot n$ time on an $n$-vertex graph when parameterized by the treewidth~$k$. Note that this implies the same result for the Grundy domination problems. Similarly as it has been done for the dual of \ZN{} in~\cite{bhyravarapu2025parameterized}, we show that this algorithm can also be used to solve the dual of \GN{} in the same time bound when $k$ is the solution size. This adaptation does not work for the other duals since the treewidth of graphs with a forcing set of size at most~$k$ is unbounded for these variants -- as are most of the reasonable graph parameters. In fact, we show that the dual problems of {\LN} and {\LLN} are \WOne-hard, leaving open only the parameterized complexity of the dual of \TN{}.

\setlength{\tabcolsep}{8pt}
\begin{table}[t]
	\centering
	\caption{Summary of the main complexity results given in this paper. Value $n$ stands for the number of vertices of the graph, value $k$ stands for the length of the dominating sequence. Value $\alpha$ stands for the \emph{independence number} of the graph, i.e., the maximal size of an independent set. Note that the \FPT{} results marked with an asterisk only hold when $\alpha$ is part of the input.}\label{tab:results}
    \resizebox{\textwidth}{!}{
	\begin{tabular}{l l l l l l l l l l l l}
		\addlinespace
		\toprule
		problem \textbackslash{} parameter & $k$ & $n - k$ & treewidth & $n-k + \alpha$ & $\alpha$ \\[0.5ex] \toprule
       \GN & \WOne-comp. & \FPT & \FPT & \FPT & \paraNP-hard  \\
           &  \footnotesize Thm.~\ref{thm:w3} & \footnotesize Cor.~\ref{corol:gn-zn}     & \footnotesize Cor.~\ref{corol:grundy-tw} & \footnotesize Cor.~\ref{corol:gn-zn} & \footnotesize \cite{bresar2023computation},~Cor.~\ref{corol:tn-cb} \\[0.5ex]\midrule
       \TN & \WOne-comp. & \XP & \FPT & \FPT\textsuperscript{*} & \paraNP-hard \\
           & \footnotesize Thm.~\ref{thm:w3}  & \footnotesize Thm.~\ref{thm:xp}  & \footnotesize Cor.~\ref{corol:grundy-tw} & \footnotesize Thm.~\ref{thm:fpt-alpha} & \footnotesize Cor.~\ref{corol:tn-cb}    \\[0.5ex]\midrule
       \ZN & \WOne-comp. & \FPT & \FPT & \FPT & \paraNP-hard \\
           & \footnotesize Thm.~\ref{thm:w3}  & \footnotesize\cite{aazami2008hardness,bhyravarapu2025parameterized}, Cor.~\ref{corol:gn-zn}              & \footnotesize \cite{aazami2008hardness,bhyravarapu2025parameterized}, Cor.~\ref{corol:grundy-tw}  & \footnotesize\cite{aazami2008hardness,bhyravarapu2025parameterized}, Cor.~\ref{corol:gn-zn} & \footnotesize Cor.~\ref{corol:tn-cb}  \\[0.5ex]\midrule
       \LN & \WOne-comp. & \WOne-hard \& \XP & \FPT & \FPT\textsuperscript{*} & \paraNP-hard \\
           & \footnotesize Thm.~\ref{thm:w3}  & \footnotesize Thms.~\ref{thm:force-w1} \& \ref{thm:xp}              & \footnotesize Cor.~\ref{corol:grundy-tw} & \footnotesize Thm.~\ref{thm:fpt-alpha} & \footnotesize Cor.~\ref{corol:tn-cb} \\[0.5ex]\midrule
       \LLN & \WOne-comp. & \WOne-hard \& \XP & \FPT & \FPT\textsuperscript{*} & \WOne-hard \& \XP\\
           & \footnotesize Thm.~\ref{thm:lln}  & \footnotesize Thms.~\ref{thm:force-w1} \& \ref{thm:xp}  & \footnotesize Cor.~\ref{corol:grundy-tw} & \footnotesize Thm.~\ref{thm:fpt-alpha} & \footnotesize Thm.~\ref{thm:lln} \& Cor.~\ref{corol:lln-alpha}  \\[0.5ex]\bottomrule \addlinespace
	\end{tabular}%
    }
\end{table}

\section{Preliminaries}

\subparagraph{General Notions}
For $k \in \N$, we write $[k]$ as shorthand for $\{1,\dots,k\}$.
All graphs considered here are finite. We denote by $V(G)$ the vertex set and by $E(G)$ the edge set of a graph $G$. The set $N(v)$ is the \emph{(open) neighborhood} of $v$ in $G$, i.e., the set of vertices $w$ for which there is an edge $vw \in E(G)$. We denote by $N[v]$ the \emph{closed neighborhood} of $G$, i.e., $N[v] = N(v) \cup \{v\}$. A \emph{hypergraph} $\H = (X, \E)$ consists of a vertex set $X$ and an edge set $\E \subseteq \cP(X)$ where $\cP(X)$ is the power set of $X$.

A graph is \emph{bipartite} if there is a partition of $V(G)$ into independent sets $A$ and $B$. A~graph is \emph{co-bipartite} if it is the complement of a bipartite graph. Equivalently, a graph is co-bipartite if $V(G)$ can be partitioned into two cliques $A$ and $B$. A graph is a \emph{split graph} if $V(G)$ can be partitioned into a clique and an independent set. A \emph{vertex cover} of a graph $G$ is a set $S \subseteq V(G)$ such that for all $vw \in E(G)$ it holds that $\{v,w\} \cap S$ is not empty.

We will also consider \emph{directed graphs}, or \emph{digraphs} for short. We will call the edges of such a graph \emph{arcs}. A digraph is \emph{acyclic} if it does not contain a directed cycle. A \emph{topological sorting} of a digraph $G$ is an ordering $(v_1, \dots, v_n)$ of its vertices such that if $(v_i,v_j)$ is an arc of $G$, then $i < j$. Note that a digraph has a topological sorting if and only if it is acyclic.

Let $\mathbb{G}$ be the family of all graphs. A \emph{graph parameter} is a mapping $\xi : \mathbb{G} \to \N$. We say that a graph parameter $\xi$ is \emph{unbounded} if for every $k \in \N$ there is a graph $G \in \mathbb{G}$ such that~$\xi(G) \geq k$.

\begin{definition}
    A \emph{tree decomposition} of a graph $G$ is a pair $(T,\{X_t\}_{t\in V(T)})$ consisting of a tree $T$ and a mapping assigning to each node $t\in V(T)$ a set $X_t\subseteq V(G)$ (called \emph{bag}) such that 
    \begin{enumerate}
    \item $\bigcup_{t \in V(T)} X_t = V(G)$,
    \item for every edge $uv \in E(G)$, there is a node $t \in V(T)$ such that $u,v \in X_t$,
    \item for every vertex $v\in V(G)$, the nodes of bags containing $v$ form a subtree of $T$.
    \end{enumerate}
    The \emph{width} of a tree decomposition is the maximum size of a bag minus~1. The \emph{treewidth} of a graph $G$ is the minimal width of a tree decomposition of $G$.
\end{definition}

The terms \emph{path decomposition} and \emph{pathwidth} are defined accordingly, where the tree $T$ is replaced by a path.

\subparagraph{Parameterized Complexity}

We give a short introduction to parameterized algorithms and complexity. For further information, we refer to the textbooks of Cygan~et~al.~\cite{cygan2015param} as well as Downey and Fellows~\cite{downey2013fundamentals}.

A \emph{parameterized problem} is a language $L \subseteq \Sigma^* \times \N$, where $\Sigma$ is a fixed, finite alphabet. For an instance $(x,k) \in \Sigma^* \times \N$, we say that $k$ is the \emph{parameter}. A parameterized problem $L \subseteq \Sigma^* \times \N$ is \emph{fixed-parameter tractable} (\FPT{}) if there is an algorithm $\A$ and a computable function $f$ such that for every instance $(x,k) \in \Sigma^* \times \N$, algorithm $\A$ correctly decides whether $(x,k) \in L$ in time bounded by $f(k) \cdot |(x,k)|^{\O(1)}$. An \emph{\FPT{} reduction} from one parameterized problem $L_1$ to another parameterized problem $L_2$ on the same alphabet $\Sigma$ is a mapping $R : \Sigma \times \N \to \Sigma \times \N$ such that for all $(x,k) \in \Sigma^* \times \N$ with $(x',k') = R((x,k))$ it holds: (1)~$(x,k) \in L_1$ if and only if $(x',k') \in L_2$,~(2) $k' \leq g(k)$ for some computable function $g$, and~(3) $(x',k')$ can be computed in $f(k) \cdot |(x,k)|^{\O(1)}$ time for some computable function $f$. We will also consider \emph{polynomial-time \FPT{} reductions}, where in (3) the term $f(k)$ is replaced by a constant. We say that two parameterized problems $L_1$ and $L_2$ are \emph{\FPT-equivalent} if there is an \FPT{} reduction from $L_1$ to $L_2$ and vice versa.

A parameterized problem $L \subseteq \Sigma^* \times \N$ is \emph{slice-wise polynomial} (\XP{}) if there is an algorithm~$\A$ and a computable function $f$ such that for every instance $(x,k) \in \Sigma^* \times \N$, algorithm~$\A$ correctly decides whether $(x,k) \in L$ in time bounded by $f(k) \cdot |(x,k)|^{f(k)}$. Unlike \FPT{}, we are allowed to have an exponent on the input size that is not constant but depends on the parameter $k$.

A parameterized problem is \emph{\paraNP-hard} if there is a value $k \in \N$ such that the problem is \NP-hard when restricted to instances $(x,k)$. This implies that no \XP{} algorithm exists for the problem, unless $\P = \NP$.

To define the \emph{$\W$-hierarchy}, we sketch the definition of \emph{Boolean circuits}. Such a Boolean circuit can be seen as a directed acyclic graph whose vertices are either AND gates or OR gates which transform the values of their (possibly negated) inputs into their conjunction or disjunction, respectively. We distinguish \emph{small gates} and \emph{large gates}, where small gates have at most four inputs and large gates have more than four inputs. Note that the value four could be replaced by any other fixed value larger than one, but four works best for our purposes. The \emph{depth} of a circuit is the maximal number of gates on a path from an input variable to the output. The \emph{weft} of a circuit is the maximal number of \emph{large} gates on such a path. We define $\C_{t,d}$ to be the class of circuits of weft at most $t$ and depth at most $d$. For some class~$\C$ of circuits, the problem WCS$[\C]$ asks whether a given circuit has a satisfying assignment of \emph{weight}~$k$, i.e., with exactly $k$ variables receiving value $\true$. A parameterized problem $L$ is in the class \WT{} for some $t \in \N$ if there is an \FPT{} reduction from $L$ to WCS[$\C_{t,d}$] for some fixed $d \geq 1$. A parameterized problem $L$ is \emph{\WT-hard} if for all problems $L' \in \WT$, there is an \FPT{} reduction from $L'$ to $L$. Since the relation of having an \FPT{} reduction is transitive, it is sufficient to reduce one \WT-hard problem to $L$. If a \WT-hard problem $L$ is also in $\WT$, then $L$ is \emph{\WT-complete}.

We will use the following parameterized problem in our reduction that was shown to be $\WOne$-hard by Pietrzak~\cite{pietrzak2003parameterized} and independently by Fellows et al.~\cite{fellows2009param}.

\begin{center}
\fbox{\parbox{0.98\textwidth}{\noindent
{\MCP}\\[.8ex]
\begin{tabular*}{.93\textwidth}{rl}
{\em Input:} & A graph $G$, $k \in \N$, a partition of $V(G)$ into $k$ independent sets $V^1, \dots, V^k$. \\
{\em Parameter:} & $k$\\
{\em Question:} & Is there a clique of size $k$ in $G$?
\end{tabular*}
}}
\end{center}

The Exponential Time Hypothesis, formulated by Impagliazzo et al.~\cite{impagliazzo2001which}, states that there is an $\epsilon > 0$ such that 3-SAT cannot be solved in $2^{\epsilon n}$ time for formulas with $n$ variables. Assuming that this hypothesis holds, one can give lower bounds on the running time of algorithms for \MCP{}.

\begin{theorem}[Cygan et al.~\cite{cygan2015param}, Lokshtanov et al.~\cite{lokshtanov2011lower}]\label{thm:lower}
Assuming the Exponential Time Hypothesis, there is no $f(k) n^{o(k)}$ time algorithm for \MCP{} for any computable function~$f$, where $n$ is the number of vertices of the graph.
\end{theorem}

\section{Grundy Domination Problems}\label{sec:grundy}

\subsection{Definitions}

We start by defining the four types of vertex sequences considered in this section.

\begin{definition}[Bre\v{s}ar et al.~\cite{bresar2017grundy,bresar2014dominating,bresar2016total}]
    Let $G$ be a graph. Let $\sigma = (v_1, \dots, v_k)$ be a sequence of pairwise different vertices of $G$. If for each $i \in [k]$
    \begin{itemize}
        \item $N[v_i] \setminus \bigcup_{j=1}^{i-1} N[v_j]$ is not empty, then $\sigma$ is a \emph{dominating sequence} of $G$;
        \item $N(v_i) \setminus \bigcup_{j=1}^{i-1} N(v_j)$ is not empty, then $\sigma$ is a \emph{total dominating sequence} of $G$;
        \item $N(v_i) \setminus \bigcup_{j=1}^{i-1} N[v_j]$ is not empty, then $\sigma$ is a \emph{$Z$-sequence} of $G$;
        \item $N[v_i] \setminus \bigcup_{j=1}^{i-1} N(v_j)$ is not empty, then $\sigma$ is an \emph{$L$-sequence} of $G$.
    \end{itemize}
\end{definition}

For each of those sequences, we say that $v_i$ \emph{footprints} the vertices in $N\langle v_i\rangle_1 \setminus \bigcup_{j=1}^{i-1} N\langle v_j\rangle_2$, where $\langle \cdot \rangle_1$ and $\langle \cdot \rangle_2$ are replaced by the respective brackets. For a sequence $(v_1, \dots, v_k)$, we will say that $v_i$ is \emph{to the left of $v_j$} if $i < j$. In this case, we may also say that $v_j$ is \emph{to the right of $v_i$.} Following the definition given by Haynes and Hedetniemi~\cite{haynes2021vertex}, we say that a sequence is \emph{connected} if its vertices induce a connected subgraph of~$G$.

As mentioned in the introduction, these types of dominating sequences are related to \emph{dominating sets} of $G$, i.e., sets $S \subseteq V(G)$ with $\bigcup_{v \in S} N[v] = V(G)$, and to \emph{total dominating sets} of $G$, i.e., sets $S \subseteq V(G)$ with $\bigcup_{v \in S} N(v) = V(G)$. These relations are given in the following observation.

\begin{observation}
    For any graph $G$, it holds that
    \begin{itemize}
        \item the vertices of a maximal dominating sequence of $G$ form a dominating set of $G$, and
        \item the vertices of a maximal $L$-sequence of $G$ form a dominating set of $G$.
    \end{itemize}
    If $G$ does not contain any isolated vertices, then
    \begin{itemize}
        \item the vertices of a maximal total dominating sequence of $G$ form a total dominating set of $G$,
        \item the vertices of a maximal $Z$-sequence of $G$ form a dominating set of $G$, and
        \item the vertices of a maximal $L$-sequence of $G$ form a total dominating set of $G$.
    \end{itemize}
\end{observation}

\begin{proof}
    Suppose we have a sequence $\sigma$ of some type and the set $S$ of vertices in $\sigma$ is not dominating. Then there is a vertex $v \in V(G) \setminus \bigcup_{u \in S}N[u] \subseteq V(G) \setminus \bigcup_{u \in S}N(u)$. Thus, this vertex can be added to the end of $\sigma$ if $\sigma$ is a dominating sequence or an $L$-sequence. If $G$ does not contain any isolated vertices, then, instead of $v$, we can add a neighbor of $v$ to the end of $\sigma$ if $\sigma$ is a $Z$-sequence. In all three cases, $\sigma$ is not maximal.

    Now suppose that $S$ is not total dominating and $G$ does not contain any isolated vertices. Then there is a vertex $v \in V(G) \setminus \bigcup_{u \in S}N(u)$. We may add a neighbor of $v$ to the end of $\sigma$ if $\sigma$ is a total dominating sequence or an $L$-sequence, again implying that $\sigma$ is not maximal.
\end{proof}

Note that the statements given in that observation cannot be improved. The total dominating sequences and $Z$-sequences of edgeless graphs are empty and, thus, do not induce dominating sets. Furthermore, edgeless graphs do not have total dominating sets, i.e., we cannot replace the dominating sets by total dominating sets in the first two statements. The dominating sequences and $Z$-sequences of complete graphs contain only one element and, thus, do not induce total dominating sets. This implies that we cannot give any statement including total dominating sets for these two types of sequences.

\bigskip
We will consider the following parameterized problem.
\begin{center}
\fbox{\parbox{0.98\textwidth}{\noindent
{\GN}\\[.8ex]
\begin{tabular*}{.93\textwidth}{rl}
{\em Input:} & A graph $G$, $k \in \N$.\\
{\em Parameter:} & $k$ \\
{\em Question:} & Is there a dominating sequence of $G$ of length at least $k$?
\end{tabular*}
}}
\end{center}

The problems \TN{}, \ZN{}, and \LN{} are defined analogously. For all four problems, we will also consider the connected variant, where we ask for a connected sequence of length at least $k$. 

Note that isolated vertices are rather uninteresting for the Grundy domination problems as they either are part of every maximal sequence (in case of dominating and $L$-sequences) or of no sequence (in case of total dominating and $Z$-sequences or some connected variant). Hence, we assume in the following that no instance contains isolated vertices.

\begin{assumption}\label{assumption}
    All instances considered here do not contain isolated vertices.
\end{assumption}

For a hypergraph $\H = (X, \E)$, we say that a sequence $(C_1, \dots, C_k)$ with $C_i \in \E$ for all $i \in [k]$ is a \emph{covering sequence} if for each $i \in [k]$ the set $C_i \setminus \bigcup_{j=1}^{i-1} C_j$ is not empty. This motivates the following parameterized problem that was introduced in \cite{bresar2014dominating}.

\begin{center}
\fbox{\parbox{0.98\textwidth}{\noindent
{\HYPER{}}\\[.8ex]
\begin{tabular*}{.93\textwidth}{rl}
{\em Input:} & A hypergraph $\H = (X,\E)$, $k \in \N$.\\
{\em Parameter:} & $k$\\
{\em Question:} & Is there a covering sequence of $\H$ of length at least $k$?
\end{tabular*}
}}
\end{center}

\subsection{An Auxiliary Problem: One-Sided Grundy Total Domination}

Instead of proving the \WOne-hardness of the four variants directly, we first prove the \WOne-hardness of the following auxiliary problem.

\begin{center}
\fbox{\parbox{0.98\textwidth}{\noindent
{\BIP}\\[.8ex]
\begin{tabular*}{.95\textwidth}{rl}
{\em Input:} & A bipartite graph $G$ with bipartition $V(G) =A \dot \cup B$, $k \in \N$.\\
{\em Parameter:} & $k$\\
{\em Question:} & Is there a total dominating sequence of length $k$ containing only vertices of $A$?
\end{tabular*}
}}
\end{center}

We will say that such a sequence is a total dominating sequence of $A$. Note that -- in difference to what the name suggests -- the vertices in $A$ will not necessarily be dominated by the sequence.
The rest of the subsection is dedicated to the proof of the following theorem.

\begin{theorem}\label{thm:bip}
   \BIP{} is \WOne-hard when it is parameterized by the solution size.
\end{theorem}

\newcommand{\sel}{\alpha}
\newcommand{\ver}{\beta}

\begin{figure}
    \centering
    \begin{tikzpicture}
\Large
\begin{scope}[scale=0.4]

\draw (0,0) rectangle (6,1);
\draw (8,0) rectangle (14,1);
\draw (4,-2) rectangle (10,-1);

\coordinate (X111) at (0.5,1);
\coordinate (X112) at (1.5,1);
\coordinate (X113) at (2.5,1);
\coordinate (X114) at (3.5,1);
\coordinate (X115) at (4.5,1);
\coordinate (X116) at (5.5,1);

\coordinate (X1q1) at (8.5,1);
\coordinate (X1q2) at (9.5,1);
\coordinate (X1q3) at (10.5,1);
\coordinate (X1q4) at (11.5,1);
\coordinate (X1q5) at (12.5,1);
\coordinate (X1q6) at (13.5,1);

\node[transform shape] at (3,0.5) {$X^1_1$};
\node[transform shape]  at (11,0.5) {$X^1_q$};
\node[transform shape]  at (7,-1.5) {$Y^1$};

\draw (0.5,0) -- (4.5,-1);
\draw (1.5,0) -- (5.5,-1);
\draw (2.5,0) -- (6.5,-1);
\draw (3.5,0) -- (7.5,-1);
\draw (4.5,0) -- (8.5,-1);
\draw (5.5,0) -- (9.5,-1);

\draw (8.5,0) -- (4.5,-1);
\draw (9.5,0) -- (5.5,-1);
\draw (10.5,0) -- (6.5,-1);
\draw (11.5,0) -- (7.5,-1);
\draw (12.5,0) -- (8.5,-1);
\draw (13.5,0) -- (9.5,-1);
\end{scope}

\begin{scope}[xshift=7cm,scale=0.4]
\draw (0,0) rectangle (6,1);
\draw (8,0) rectangle (14,1);
\draw (4,-2) rectangle (10,-1);

\node[transform shape]  at (3,0.5) {$X^k_1$};
\node[transform shape]  at (11,0.5) {$X^k_q$};
\node[transform shape]  at (7,-1.5) {$Y^k$};

\coordinate (Xk11) at (0.5,1);
\coordinate (Xk12) at (1.5,1);
\coordinate (Xk13) at (2.5,1);
\coordinate (Xk14) at (3.5,1);
\coordinate (Xk15) at (4.5,1);
\coordinate (Xk16) at (5.5,1);

\coordinate (Xkq1) at (8.5,1);
\coordinate (Xkq2) at (9.5,1);
\coordinate (Xkq3) at (10.5,1);
\coordinate (Xkq4) at (11.5,1);
\coordinate (Xkq5) at (12.5,1);
\coordinate (Xkq6) at (13.5,1);

\draw (0.5,0) -- (4.5,-1);
\draw (1.5,0) -- (5.5,-1);
\draw (2.5,0) -- (6.5,-1);
\draw (3.5,0) -- (7.5,-1);
\draw (4.5,0) -- (8.5,-1);
\draw (5.5,0) -- (9.5,-1);

\draw (8.5,0) -- (4.5,-1);
\draw (9.5,0) -- (5.5,-1);
\draw (10.5,0) -- (6.5,-1);
\draw (11.5,0) -- (7.5,-1);
\draw (12.5,0) -- (8.5,-1);
\draw (13.5,0) -- (9.5,-1);
\end{scope}

\begin{scope}[xshift=0.15cm,yshift=1cm]
\begin{scope}[scale=0.4, every label/.append style={transform shape}]
\draw (0,0) rectangle (6,1);

%\draw (-0.35,-0.35) rectangle (6.35,2.35);

\node[draw, inner sep=1pt, circle, label={[label distance=0.0cm]90:$c^{12}(1)$}] (c121) at (3,2) {};

\node[transform shape] at (3,0.5) {$W^{12}(1)$};

\node[transform shape] at (6.55,0.5) {$\dots$};
\node[transform shape] at (6.55,-2) {$\dots$};

%\draw (c121) -- (0.5,1);
%\draw (c121) -- (1.5,1);
%\draw (c121) -- (2.5,1);
\draw[all] (c121) -- (3,1);
%\draw (c121) -- (3.5,1);
%\draw (c121) -- (4.5,1);
%\draw (c121) -- (5.5,1);

\draw[dashed]  (2,0) -- (X111);
\draw[dashed]  (2,0) -- (X112);
\draw[dashed]  (2,0) -- (X113);
\draw[dashed]  (2,0) -- (X114);
\draw[dashed]  (2,0) -- (X115);
\draw[dashed]  (2,0) -- (X116);

\draw[dashed]  (5,0) -- (X1q1);
\draw[dashed]  (5,0) -- (X1q2);
\draw[dashed]  (5,0) -- (X1q3);
\draw[dashed]  (5,0) -- (X1q4);
\draw[dashed]  (5,0) -- (X1q5);
\draw[dashed]  (5,0) -- (X1q6);
\end{scope}

\begin{scope}[scale=0.4, xshift=7cm,every label/.append style={transform shape}]
\draw (0,0) rectangle (6,1);

%\draw (-0.35,-0.35) rectangle (6.35,2.35);

\node[draw, inner sep=1pt, circle, label={[label distance=0.0cm]175:$c^{12}(\ver)$}] (c12b) at (3,2) {};

\node[transform shape] at (3,0.5) {$W^{12}(\ver)$};

%\draw (c12b) -- (0.5,1);
%\draw (c12b) -- (1.5,1);
%\draw (c12b) -- (2.5,1);
\draw[all] (c12b) -- (3,1);
%\draw (c12b) -- (3.5,1);
%\draw (c12b) -- (4.5,1);
%\draw (c12b) -- (5.5,1);

\draw[dashed]  (2,0) -- (X111);
\draw[dashed]  (2,0) -- (X112);
\draw[dashed]  (2,0) -- (X113);
\draw[dashed]  (2,0) -- (X114);
\draw[dashed]  (2,0) -- (X115);
\draw[dashed]  (2,0) -- (X116);

\draw[dashed]  (5,0) -- (X1q1);
\draw[dashed]  (5,0) -- (X1q2);
\draw[dashed]  (5,0) -- (X1q3);
\draw[dashed]  (5,0) -- (X1q4);
\draw[dashed]  (5,0) -- (X1q5);
\draw[dashed]  (5,0) -- (X1q6);

\end{scope}
\end{scope}

\begin{scope}[xshift=7.15cm,yshift=1cm]
\begin{scope}[scale=0.4, every label/.append style={transform shape}]
\draw (0,0) rectangle (6,1);

%\draw (-0.35,-0.35) rectangle (6.35,2.35);

\node[draw, inner sep=1pt, circle, label={[label distance=0.0cm]5:$c^{k-1k}(1)$}] (ckk1) at (3,2) {};

\node[transform shape] at (3,0.5) {$W^{k-1k}(1)$};

\node[transform shape] at (6.55,0.5) {$\dots$};
\node[transform shape] at (6.55,-2) {$\dots$};

%\draw (ckk1) -- (0.5,1);
%\draw (ckk1) -- (1.5,1);
%\draw (ckk1) -- (2.5,1);
\draw[all] (ckk1) -- (3,1);
%\draw (ckk1) -- (3.5,1);
%\draw (ckk1) -- (4.5,1);
%\draw (ckk1) -- (5.5,1);

\draw[dashed]  (1,0) -- (Xk11);
\draw[dashed]  (1,0) -- (Xk12);
\draw[dashed]  (1,0) -- (Xk13);
\draw[dashed]  (1,0) -- (Xk14);
\draw[dashed]  (1,0) -- (Xk15);
\draw[dashed]  (1,0) -- (Xk16);

\draw[dashed]  (4,0) -- (Xkq1);
\draw[dashed]  (4,0) -- (Xkq2);
\draw[dashed]  (4,0) -- (Xkq3);
\draw[dashed]  (4,0) -- (Xkq4);
\draw[dashed]  (4,0) -- (Xkq5);
\draw[dashed]  (4,0) -- (Xkq6);
\end{scope}

\begin{scope}[scale=0.4, xshift=7cm,every label/.append style={transform shape}]
\draw (0,0) rectangle (6,1);

%\draw (-0.35,-0.35) rectangle (6.35,2.35);

\node[draw, inner sep=1pt, circle, label={[label distance=0.0cm]5:$c^{k-1k}(\ver)$}] (ckkb) at (3,2) {};

\node[transform shape] at (3,0.5) {$W^{k-1}k(\ver)$};

%\draw (ckkb) -- (0.5,1);
%\draw (ckkb) -- (1.5,1);
%\draw (ckkb) -- (2.5,1);
\draw[all] (ckkb) -- (3,1);
%\draw (ckkb) -- (3.5,1);
%\draw (ckkb) -- (4.5,1);
%\draw (ckkb) -- (5.5,1);

\draw[dashed]  (1,0) -- (Xk11);
\draw[dashed]  (1,0) -- (Xk12);
\draw[dashed]  (1,0) -- (Xk13);
\draw[dashed]  (1,0) -- (Xk14);
\draw[dashed]  (1,0) -- (Xk15);
\draw[dashed]  (1,0) -- (Xk16);

\draw[dashed]  (4,0) -- (Xkq1);
\draw[dashed]  (4,0) -- (Xkq2);
\draw[dashed]  (4,0) -- (Xkq3);
\draw[dashed]  (4,0) -- (Xkq4);
\draw[dashed]  (4,0) -- (Xkq5);
\draw[dashed]  (4,0) -- (Xkq6);

\end{scope}
\end{scope}

\node at (6.3,0.5) {$\dots$};
\draw (4.8,4) rectangle (7.8,4.4);
\begin{scope}[xshift=6.3cm, yshift=4.2cm, scale=0.4]
\node[transform shape] at (0,0) {$Y = Y^1 \cup \dots \cup Y^k$};
\end{scope}
\coordinate (Y1) at (5.4,4);
%\coordinate (Y2) at (5,2.5);
\coordinate (Y3) at (6.2,4);
%\coordinate (Y4) at (5.8,2.5);
\coordinate (Y5) at (6.9,4);
%\coordinate (Y6) at (6.6,2.5);

\node[draw, inner sep =1pt, circle] (f) at (6.2,3) {};
\node[draw, inner sep =1pt, circle] (g) at (6.2,2.7) {};
\begin{scope}[xshift=6.35cm,yshift=3cm,scale=0.4]
\node[transform shape] at (0,0) {$f$};
\end{scope}
\begin{scope}[xshift=6.2cm,yshift=2.5cm,scale=0.4]
\node[transform shape] at (0,0) {$g$};
\end{scope}

%\draw (f) -- (Y1);
\draw[all] (f) -- (Y3);
%\draw (f) -- (Y5);
\draw (f) -- (g);

\draw[all] (c121) -- (Y1);
%\draw (c121) -- (Y2);
%\draw (c121) -- (Y3);
%\draw (c121) -- (Y4);
%\draw (c121) -- (Y5);
\draw (c121) -- (g);
%\draw (c121) -- (Y6);

\draw[all] (c12b) -- (Y1);
%\draw (c12b) -- (Y2);
%\draw (c12b) -- (Y3);
%\draw (c12b) -- (Y4);
%\draw (c12b) -- (Y5);
\draw (c12b) -- (g);
%\draw (c12b) -- (Y6);

%\draw (ckk1) -- (Y1);
%\draw (ckk1) -- (Y2);
%\draw (ckk1) -- (Y3);
%\draw (ckk1) -- (Y4);
\draw[all] (ckk1) -- (Y5);
\draw (ckk1) -- (g);
%\draw (ckk1) -- (Y6);

%\draw (ckkb) -- (Y1);
%\draw (ckkb) -- (Y2);
%\draw (ckkb) -- (Y3);
%\draw (ckkb) -- (Y4);
\draw[all] (ckkb) -- (Y5);
\draw (ckkb) -- (g);
%\draw (ckkb) -- (Y6);

%\draw (c121) -- ++(-1.25,0) -- ++ (0,-2.05) -- (Y11);
%\draw (c12b) -- ++(-0.75,0.375) -- ++(-2.75,0) -- ++ (0,-2.55) -- (Y11);
\end{tikzpicture}
    \caption{An overview of the construction of the proof of \cref{thm:bip}. For better clarity, the vertices of $Y^1$ and $Y^k$ appear twice in the picture. Thick edges imply that all edges between the vertex and the vertex set are present in the graph. Between the vertices of $W^{ij}(b)$ and the $X^i_p$, we only sketch some of the non-edges by dashed lines to suggest that most edges are present.}
    \label{fig:overview}
\end{figure}

\begin{figure}
    \centering
    \begin{tikzpicture}[scale=0.9125]
\scriptsize
\begin{scope}[xshift=-0.25cm]
\begin{scope}[scale=0.7]
    \node[svertex, label={[name=li11] 90:$x^1_1(1)$}] (xi11) at (0,0) {};
   % \node[svertex, label={[name=li12] 90:$x^1_1(2)$}] (xi12) at (1,0) {};
    \node at (0.75,0) {$\dots$};
    \node[svertex, label={[name=li1a] 90:$x^1_1(\sel)$}] (xi1a) at (1.5,0) {};
\end{scope}

   \node[label={[name=li1] 180:$X^1_1$}, rectangle, rounded corners, draw=lipicsBulletGray, inner sep=2pt, fit={(xi11) (xi1a) (li11) (li1a)}] (Xi1) {};
    \node at (1.9,0.25) {\dots};

\begin{scope}[xshift=2.75cm,scale=0.7]
    \node[svertex, label={[name=liq1] 90:$x^1_q(1)$}] (xiq1) at (0,0) {};
    %\node[svertex, label={[name=liq2] 90:$x^1_q(2)$}] (xiq2) at (1,0) {};
    \node at (0.75,0) {$\dots$};
    \node[svertex, label={[name=liqa] 90:$x^1_q(\sel)$}] (xiqa) at (1.5,0) {};
\end{scope}

    \node[label={[name=liq] 0:$X^1_q$}, rectangle, rounded corners, draw=lipicsBulletGray, inner sep=2pt, fit={(xiq1) (xiqa) (liq1) (liqa)}] (Xiq) {};

\begin{scope}[xshift=1.375cm,yshift=-0.75cm,scale=0.7]
    \node[svertex, label={[name=lyi1] 180:$y^1(1)$}] (yi1) at (0,0) {};
    %\node[svertex, label={[name=lyi2] 270:$y^1(2)$}] (yi2) at (1,0) {};
    \node at (0.75,0) {$\dots$};
    \node[svertex, label={[name=lyia] 0:$y^1(\sel)$}] (yia) at (1.5,0) {};
\end{scope}

\draw (xi11) -- (yi1);
%\draw (xi12) -- (yi2);
\draw (xi1a) -- (yia);

\draw (xiq1) -- (yi1);
%\draw (xiq2) -- (yi2);
\draw (xiqa) -- (yia);

\node[label={170:$\S^1$}, rectangle, rounded corners, draw=lipicsBulletGray, inner sep=2pt, fit={(Xi1) (Xiq) (li1) (liq) (li1a) (lyi1) (lyia)}] (Si) {};
\end{scope}

\node at (5.3,-0.2) {\dots};

\begin{scope}[xshift=7cm]
\begin{scope}[scale=0.7]
    \node[svertex, label={[name=li11] 90:$x^k_1(1)$}] (xi11) at (0,0) {};
   % \node[svertex, label={[name=li12] 90:$x^k_1(2)$}] (xi12) at (1,0) {};
    \node at (0.75,0) {$\dots$};
    \node[svertex, label={[name=li1a] 90:$x^k_1(\sel)$}] (xi1a) at (1.5,0) {};
\end{scope}

   \node[label={[name=li1] 180:$X^k_1$}, rectangle, rounded corners, draw=lipicsBulletGray, inner sep=2pt, fit={(xi11) (xi1a) (li11) (li1a)}] (Xi1) {};
    \node at (1.9,0.25) {\dots};

\begin{scope}[xshift=2.75cm,scale=0.7]
    \node[svertex, label={[name=liq1] 90:$x^k_q(1)$}] (xiq1) at (0,0) {};
%    \node[svertex, label={[name=liq2] 90:$x^k_q(2)$}] (xiq2) at (1,0) {};
    \node at (0.75,0) {$\dots$};
    \node[svertex, label={[name=liqa] 90:$x^k_q(\sel)$}] (xiqa) at (1.5,0) {};
\end{scope}

    \node[label={[name=liq] 0:$X^k_q$}, rectangle, rounded corners, draw=lipicsBulletGray, inner sep=2pt, fit={(xiq1) (xiqa) (liq1) (liqa)}] (Xiq) {};

\begin{scope}[xshift=1.375cm,yshift=-0.75cm,scale=0.7]
    \node[svertex, label={[name=lyi1] 180:$y^k(1)$}] (yi1) at (0,0) {};
   % \node[svertex, label={[name=lyi2] 270:$y^k(2)$}] (yi2) at (1,0) {};
    \node at (0.75,0) {$\dots$};
    \node[svertex, label={[name=lyia] 0:$y^k(\sel)$}] (yia) at (1.5,0) {};
\end{scope}

\draw (xi11) -- (yi1);
\draw (xi1a) -- (yia);

\draw (xiq1) -- (yi1);
\draw (xiqa) -- (yia);

\node[label={10:$\S^k$}, rectangle, rounded corners, draw=lipicsBulletGray, inner sep=2pt, fit={(Xi1) (Xiq) (li1) (liq) (li1a) (lyi1) (lyia)}] (Si) {};
\end{scope}
\end{tikzpicture}
    \caption{The selection gadgets of the proof of \cref{thm:bip}.}
    \label{fig:sel}
\end{figure}

We reduce from \MCP{}. Let $G$ be the input graph and $V^1, \dots, V^k$ be the partition of $V(G)$ into independent sets. Without loss of generality, we may assume that $k \geq 2$ and that all color classes are of the same size, i.e., we assume that $V^i = \{v^i_1, \dots, v^i_q\}$ for $i \in [k]$ and some fixed value $q \in \N$. To simplify the notation, we define $\sel := \ver := 2k + 1$. We build the graph $G'$ as follows (see \cref{fig:overview} for an overview):

\begin{description}
    \item[Selection Gadgets] For every $i \in [k]$, we have a selection gadget $\S^i$ that contains for every $p \in [q]$ the vertex set $X^i_p := \{x^i_p(a) \mid a \in [\alpha]\}$. We call these vertices \emph{selection vertices}. We define $X^i := \bigcup_{p=1}^{q} X^i_p$. Furthermore, $\S^i$ contains the vertices $Y^i = \{y^i(1), \dots, y^i(\sel)\}$. For every $a \in [\sel]$ and every $p \in [q]$, the vertex $x^i_p(a)$ is adjacent to $y^i(a)$. We define $Y := \bigcup_{i=1}^k Y^i$ (see \cref{fig:sel} for an illustration).
    \item[Verification Gadgets] For every set $\{i,j\} \subseteq [k]$, we have $\ver$ gadgets $\C^{ij}(1), \dots, \C^{ij}(\ver)$. Such a verification gadget $\C^{ij}(b)$ contains for every edge $v^i_pv^j_r \in E(G)$ an \emph{edge vertex} $w^{ij}_{pr}(b)$. The set $W^{ij}(b)$ is defined as $W^{ij}(b) := \{w^{ij}_{pr}(b) \mid v^i_pv^j_r \in E(G)\}$. Furthermore, $\C^{ij}(b)$ contains a \emph{verification vertex} $c^{ij}(b)$. The vertex $c^{ij}(b)$ is adjacent to all edge vertices in $W^{ij}(b)$. 
    Note that we can write $\C^{ij}(b)$ and $\C^{ji}(b)$ interchangeably. Similarly, the vertices $c^{ij}(b)$ and $c^{ji}(b)$ as well as the vertices $w^{ij}_{pr}$ and $w^{ji}_{rp}$ are identical. This means that the indices $i$ and $j$ are bounded to the indices $p$ and $r$ by their order, i.e., the first (second) index in superscript belongs to the first (second) index in subscript. (see \cref{fig:ver} for an illustration). 
    \item[Blocker Gadget] The gadget contains the two adjacent vertices $f$ and $g$.
\end{description}

Besides the edges within the gadgets, there are also edges between vertices of different gadgets (see \cref{fig:overview,fig:ver}).

\begin{enumerate}[(E1)]
    \item Vertex $f$ is adjacent to all vertices in $Y$.
    \item Every verification vertex is adjacent to $g$ as well as to all vertices in $Y$.
    \item Edge vertex $w^{ij}_{pr}(b)$ is adjacent to all vertices in $X^i \cup X^j$ except for those in $X_p^i$ and~$X_r^j$.\label{item:bfs-w1-edge1}
\end{enumerate}

The graph $G'$ is a bipartite graph. The bipartition consists of the sets $A$ and $B$, where $A = \bigcup_{i=1}^{k}\big(X^i\big) \cup \bigcup_{i,j \in[k], i\neq j} \bigcup_{b=1}^\ver \big(c^{ij}(b)\big) \cup \{f\}$ and $B = V(G') \setminus A$.

\begin{figure}
    \centering
    \begin{tikzpicture}
\scriptsize
    \begin{scope}[xshift=-1.9cm]
        \node[svertex, label={[name=lw1] -180:$w^{ij}_{pt}(b)$}] (w1) at (0,0) {};
        \node[svertex, label={[name=lw2, label distance=2pt] -2:$w^{ij}_{pv}(b)$}] (w2) at (2,0) {};
        \node[svertex, label={[name=lw3] -0:$w^{ij}_{sx}(b)$}] (w3) at (4,0) {};
    \end{scope}

    \begin{scope}[xshift=-4cm, yshift=1.5cm]
        \node[svertex, label={[name=lxi1] -270:$x^{i}_{p}(a_1)$}] (xi1) at (0,0) {};
        \node[svertex, label={[name=lxi2] -270:$x^{i}_{r}(a_2)$}] (xi2) at (1.5,0) {};
        \node[svertex, label={[name=lxi3] -270:$x^{i}_{s}(a_3)$}] (xi3) at (3,0) {};

        \node[label={[name=li1] 180:$\S^i \supset$}, rectangle, rounded corners, draw=lipicsBulletGray, inner sep=2pt, fit={(xi1) (xi3) (lxi1) (lxi3)}] (Xi1) {};
    \end{scope}

    \begin{scope}[xshift=1cm, yshift=1.5cm]
        \node[svertex, label={[name=lxj1] -270:$x^{j}_{t}(a_4)$}] (xj1) at (0,0) {};
        \node[svertex, label={[name=lxj2] -270:$x^{j}_{u}(a_5)$}] (xj2) at (1.5,0) {};
        \node[svertex, label={[name=lxj3] -270:$x^{j}_{v}(a_6)$}] (xj3) at (3,0) {};

        \node[label={[name=lj1] 0:$\subset \S^j$}, rectangle, rounded corners, draw=lipicsBulletGray, inner sep=2pt, fit={(xj1) (xj3) (lxj1) (lxj3)}] (Xj1) {};
    \end{scope}

    \draw[draw=black,dashed] (xi1) -- (w1);
    \draw[draw=black,dashed] (xi1) -- (w2);
    \draw[draw=black] (xi1) -- (w3);
    \draw[draw=black] (xi2) -- (w1);
    \draw[draw=black] (xi2) -- (w2);
    \draw[draw=black] (xi2) -- (w3);
    \draw[draw=black] (xi3) -- (w1);
    \draw[draw=black] (xi3) -- (w2);
    \draw[draw=black,dashed] (xi3) -- (w3);

    \draw[draw=black,dashed] (xj1) -- (w1);
    \draw[draw=black] (xj1) -- (w2);
    \draw[draw=black] (xj1) -- (w3);
    \draw[draw=black] (xj2) -- (w1);
    \draw[draw=black] (xj2) -- (w2);
    \draw[draw=black] (xj2) -- (w3);
    \draw[draw=black] (xj3) -- (w1);
    \draw[draw=black,dashed] (xj3) -- (w2);
    \draw[draw=black] (xj3) -- (w3);

    \node[svertex, label={[name=lcij] -20:$c^{ij}(b)$}] (cij) at (-1,-1) {};

    \node[svertex, label={[name=lf] 180:$f$}] (f) at (-5.5,1.5) {};
    \node[svertex, label={[name=lg] 180:$g$}] (g) at (-4,0) {};

    \begin{scope}[xshift=-6.5cm, yshift=-1cm]

        \node (Y1) at (0,0.15) {};
        \node (Y2) at (0,0) {};
        \node (Y3) at (0,-0.15) {};

        \draw[draw=black] (cij) -- (Y1);
        \draw[draw=black] (cij) -- (Y2);
        \draw[draw=black] (cij) -- (Y3);

        \node (Y1) at (1,0) {};
        \node (Y2) at (0.5,0) {};
        \node (Y3) at (0,0) {};
        
        \draw[draw=black] (f) -- (g);
        \draw[draw=black] (f) -- (Y1);
        \draw[draw=black] (f) -- (Y2);
        \draw[draw=black] (f) -- (Y3);
        
        \draw[draw=lipicsBulletGray, rounded corners, fill=white] (-0.5,0.25) rectangle (1.5,-0.25);

        \node (Y4) at (0.5,0) {$Y$};
    \end{scope}

    \draw[draw=black] (cij) -- (w1);
    \draw[draw=black] (cij) -- (w2);
    \draw[draw=black] (cij) -- (w3);
    \draw[draw=black] (cij) -- (g);

    \node[label={[name=li1] 0:$\subset \C^{ij} (b)$}, rectangle, rounded corners, draw=lipicsBulletGray, inner sep=2pt, fit={(cij) (w1) (w2) (w3) (lw1) (lw2) (lw3) (lcij)}] (Cij) {};
\end{tikzpicture}
    \caption{One verification gadget of the proof of \cref{thm:bip} and its connection to some of the other vertices. We assume that the values $p$, $r$, $s$, $t$, $u$, $v$, and $x$ are pairwise different. Non-existing edges between selection vertices and edge vertices are drawn as dashed lines. Thick edges to $Y$ imply that the respective vertex is adjacent to all elements of $Y$. Note that the edges between $Y$ and the selection vertices are missing as they are depicted in \cref{fig:sel}.}
    \label{fig:ver}
\end{figure}
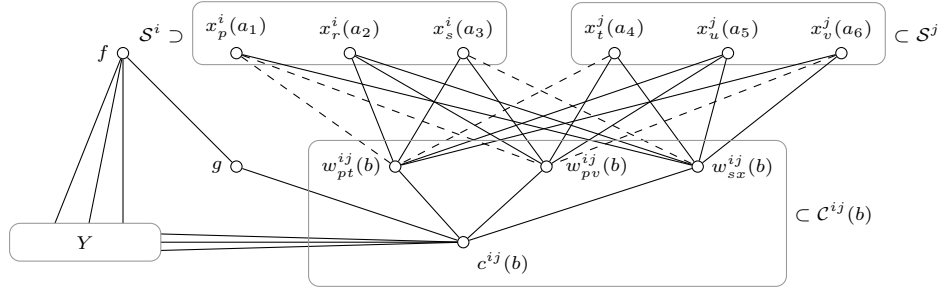

In the following, we will prove that there is a multicolored clique in $G$ of size $k$ if and only if there is a total dominating sequence of $A$ with length $\Gamma = \sel \cdot k + \ver \cdot \binom{k}{2} + 1$. 
The idea of the construction is the following. The selection vertices in the total dominating sequence represent a choice of one vertex per color class of $G$ that should form a multicolored clique. This choice is done at the beginning during the \emph{selection phase}. This selection phase is finalized by adding the vertex $f$ to the sequence. Then the \emph{verification phase} starts where the goal is to visit all the verification vertices. We will prove that this is possible if and only if our selection phase has chosen the representatives of a multicolored clique.

In the first step of the proof of correctness of our construction we show that a multicolored clique in $G$ implies a total dominating sequence of $A$ of length~$\Gamma$.

\begin{lemma}\label{lemma:w1-hin}
    If there is a multicolored clique $\{v^1_{p_1}, \dots, v^k_{p_k}\}$ in $G$, then there is a total dominating sequence of $A$ of length $\Gamma$.
\end{lemma}

\begin{proof}
    We build the sequence $\sigma$ of $A$ as follows. For every $i \in [k]$, we take the vertices $x^i_{p_i}(1), \dots, x^i_{p_i}(\sel)$ in that order. Next we take vertex $f$. Now we traverse the verification vertices in arbitrary order. It is obvious that $\sigma$ has length $\Gamma = \sel \cdot k + \ver \cdot \binom{k}{2} + 1$. It remains to show that all elements of $\sigma$ footprint some vertex of~$B$.
    \begin{itemize}
        \item For every $a \in [\sel]$, the vertex $x^i_{p_i}(a)$ footprints vertex $y^i(a)$.
        \item Vertex $f$ footprints vertex $g$.
        \item Consider the vertex $c^{ij}(b)$. Since $v^i_{p_i}$ and $v^j_{p_j}$ are adjacent in $G$, there is a vertex $w^{ij}_{p_ip_j}(b)$ in $G'$ and this vertex is neither adjacent to $x^i_{p_i}(a)$ nor to $x^j_{p_j}(a)$ for all $a \in [\alpha]$, due to (E\ref{item:bfs-w1-edge1}). Therefore, $c^{ij}(b)$ footprints $w^{ij}_{p_ip_j}(b)$.
    \end{itemize}
    Hence, $\sigma$ is a total dominating sequence of $A$.
\end{proof}

It remains to show that a total dominating sequence of $A$ of length $\Gamma$ implies the existence of a multicolored clique of size $k$ in $G$. We will prove this using several lemmas. We assume in the following that $\sigma$ is a total dominating sequence of $A$ of length at least~$\Gamma$.

\begin{lemma}\label{lemma:bfs-w1-footprint}
    For every $i \in [k]$, there are at most two vertices of $X^i$ in $\sigma$ that footprint edge vertices.
\end{lemma}

\begin{proof}
    For each $a, a' \in [\sel]$, the vertices $x^i_p(a)$ and $x^i_p(a')$ have the same neighborhood in the set of edge vertices. Thus, there can be at most one $a \in [\sel]$ such that $x^i_p(a)$ footprints some edge vertex.

    Let $x^i_p(a_p)$, $x^i_r(a_r)$ and $x^i_s(a_s)$ be three vertices of $X^i$ such that $p \neq r$, $x^i_p(a_p)$ is to the left of $x^i_r(a_r)$, and $x^i_r(a_r)$ is to the left of $x^i_s(a_s)$. We claim that $x^i_s(a_s)$ does not footprint an edge vertex. Let $w^{ij}_{tu}(b)$ be some edge vertex. If $t \neq p$, then $w^{ij}_{tu}(b)$ is adjacent to $x^i_p(a_p)$. If $t = p$, then $w^{ij}_{tu}(b)$ is adjacent to $x^i_r(a_r)$. Therefore, $x^i_s(a)$ does not footprint any vertex~$w^{ij}_{tu}(b)$.
\end{proof}

This implies the following upper bound on the length of $\sigma$.

\begin{lemma}\label{corol:bfs-w1-maxium}
    The sequence $\sigma$ contains at most
    \begin{itemize}
        \item $\alpha + 2$ vertices of $X^i$ for all $i \in [k]$
        \item $\beta$ vertices of $\C^{ij}$ for all $\{i,j\} \subseteq [k]$.
    \end{itemize}
    Therefore, the length of $\sigma$ is at most $(\sel + 2) \cdot k + \ver \cdot \binom{k}{2} + 1 = \Gamma + 2k$.
\end{lemma}

\begin{proof}
    Due to \cref{lemma:bfs-w1-footprint}, all but two vertices of $X^i$ must footprint some vertex $y^i(a)$. As there are only $\sel$ vertices $y^i(a)$, there can be at most $\sel + 2$ vertices of $X^i$. For every $\{i,j\} \subseteq [k]$, it holds that $|\C^{ij} \cap A| = \ver$. Thus, the bound on the verification vertices is trivial. Summing up, there are at most $(\sel + 2) \cdot k$ selection vertices and at most $\ver \cdot \binom{k}{2}$ verification vertices in $\sigma$. The only remaining additional vertex is $f$. Thus, the given bound holds.
\end{proof}

Furthermore, for every $i,j \in [k]$ with $i\neq j$ there is at least one verification vertex in $\sigma$.

\begin{lemma}\label{lemma:bfs-w1-ver}
    For all $i,j \in [k]$ with $i \neq j$, there is some $b \in [\ver]$ such that $c^{ij}(b)$ is in~$\sigma$.
\end{lemma}

\begin{proof}
    Assume for contradiction that this is not the case. Then, due to \cref{corol:bfs-w1-maxium}, $\sigma$ contains at most $\Gamma + 2k - \ver = \Gamma - 1$ vertices; a contradiction to the fact that $\sigma$ has length~$\Gamma$.
\end{proof}

Next, we show that vertex $f$ can be assumed to be to the left of all verification vertices.

\begin{lemma}\label{lemma:bfs-w1-f-left}
    There is a total dominating sequence $\sigma'$ of $A$ which has the same length as $\sigma$ such that $f$ is to the left of all verification vertices in $\sigma'$.
\end{lemma}

\begin{proof}
   Let $c^{ij}(b)$ be the leftmost verification vertex in $\sigma$. Due to \cref{lemma:bfs-w1-ver}, this vertex exists. Assume for contradiction that $f$ is not to the left of $c^{ij}(b)$ in $\sigma$. Then $f$ is not part of $\sigma$ since all its neighbors are also neighbors of $c^{ij}(b)$. We create $\sigma'$ from $\sigma$ by replacing $c^{ij}(b)$ with $f$. Then $f$ footprints $g$ since it is the leftmost neighbor of $g$ in $\sigma'$. All vertices to the left of $f$ still footprint the vertices in $\sigma'$ that they footprinted in $\sigma$. In $\sigma$, vertex $c^{ij}(b)$ has footprinted all vertices of $Y \cup \{g\}$ that have not already been footprinted by some vertex to the left of $c^{ij}(b)$. Hence, all vertices to the right of $c^{ij}(b)$ in $\sigma$ only footprinted edge vertices. As $f$ is not adjacent to any edge vertex, these vertices still footprint the same vertices in $\sigma'$ as in~$\sigma$. Therefore, $\sigma'$ is a total dominating sequence of $A$ with the same length as $\sigma$.
\end{proof}

Due to this lemma, we may assume in the following that vertex $f$ is to the left of all verification vertices in $\sigma$. We now show that the selection phase really selects one vertex of $G$ per color class.

\begin{lemma}\label{lemma:bfs-w1-exactly}
    For every $i \in [k]$, there exists a unique $p_i \in [q]$ such that there is a vertex of $X^i_{p_i}$ to the left of $f$ in $\sigma$.
\end{lemma}

\begin{proof}
    First we show that there is at least one vertex of $X^i$ to the left of $f$ in $\sigma$. Assume for contradiction that this is not the case. As $f$ is adjacent to all vertices of $Y$, the vertices of $X^i$ can only footprint edge vertices. By \cref{lemma:bfs-w1-footprint}, there are at most two such vertices. Due to \cref{corol:bfs-w1-maxium}, the length of $\sigma$ is at most $(\sel + 2) \cdot (k-1) + 2 + \ver \cdot \binom{k}{2} + 1 = \Gamma + 2k - \sel = \Gamma - 1$; a contradiction as $\sigma$ has length $\Gamma$.
    
    Thus, let $x^i_{p_i}(a)$ be a vertex to the left of $f$ in $\sigma$. Assume for contradiction that there is some index $r \in [q]$ with $r \neq p_i$ such that $x^i_r(a')$ is also to the left of $f$ for some $a' \in [\sel]$. Let $j \in [k]$ be arbitrary with $i \neq j$. Due to \cref{lemma:bfs-w1-ver}, there is a vertex $c^{ij}(b)$ in $\sigma$. Due to \cref{lemma:bfs-w1-f-left}, we may assume that vertex $f$ is to the left of $c^{ij}(b)$ and, thus, $c^{ij}(b)$ does not footprint any vertex of $Y \cup \{g\}$. Therefore, $c^{ij}(b)$ footprints some edge vertex $w^{ij}_{st}(b)$. As $x^i_{p_i}(a)$ is not adjacent to $w^{ij}_{st}(b)$, it holds that $s = p_i$, due to (E\ref{item:bfs-w1-edge1}). However, then $x^i_r(a')$ is adjacent to $w^{ij}_{st}(b)$ and $c^{ij}(b)$ cannot footprint $w^{ij}_{st}(b)$; a contradiction.
\end{proof}

This allows us to finalize the proof of the second direction of \cref{thm:bip}.

\begin{lemma}\label{lemma:w1-rueck}
    Let $p_1, \dots, p_k$ be values in $[q]$ such that for each $i \in [k]$ there is a vertex of $X^i_{p_i}$ to the left of $f$ in $\sigma$. Then the set $\{v^1_{p_1}, \dots, v^k_{p_k}\}$ induces a clique in~$G$.
\end{lemma}

\begin{proof}
    Assume for contradiction that $v^i_{p_i}$ is not adjacent to $v^j_{p_j}$. Let $x^{i}_{p_i}(a)$ and $x^{j}_{p_j}(a')$ be two vertices that are to the left of $f$ in $\sigma$. Note that, due to \cref{lemma:bfs-w1-f-left}, we may assume that these vertices are also to the left of all verification vertices. Since $v^i_{p_i}$ is not adjacent to $v^j_{p_j}$ in $G$, there are no edge vertices $w^{ij}_{p_ip_j}(b)$ in $G'$ for every $b \in [\ver]$. All other edge vertices $w^{ij}_{rs}(b)$ are adjacent to $x^i_{p_i}(a)$ or $x^j_{p_j}(a')$, due to (E\ref{item:bfs-w1-edge1}). Hence, no verification vertex $c^{ij}(b)$ can be an element of $\sigma$. This contradicts \cref{lemma:bfs-w1-ver}.
\end{proof}

\Cref{lemma:w1-hin} as well as \cref{lemma:bfs-w1-exactly,lemma:w1-rueck} imply that the graph $G'$ is a proper reduction from \MCP{} to \BIP{}. Since $G'$ can be constructed in polynomial-time and since $\Gamma \in \O(k^3)$, the graph $G'$ forms a polynomial-time \FPT{} reduction. This implies the correctness of \cref{thm:bip}.

\subsection{W[1]-Completeness of Grundy Domination Problems}

First, we will use \cref{thm:bip} to show \WOne-hardness for the four Grundy domination problems and their connected variants. We start with the following lemma.

\begin{lemma}[Lin~{\cite[Proposition 4.2]{lin2019zero}}]\label{lemma:A-B}
    Let $G$ be a bipartite graph with bipartition $V(G) = A \dot\cup B$. Then the following statements are equivalent:
    \begin{enumerate}[(i)]
        \item There is a total dominating sequence of $A$ with $k$ elements.
        \item There is a total dominating sequence of $B$ with $k$ elements.
        \item There is a total dominating sequence of $G$ with $2k$ elements.
    \end{enumerate}
\end{lemma}

This lemma implies directly that there is a trivial \FPT{} reduction from \BIP{} to \TN{} on bipartite graphs and vice versa. Besides this, we observe that the polynomial-time reduction from \TN{} to \TN{} on split graphs given in~\cite{bresar2018total} is in fact an \FPT{} reduction. Thus, the following holds.

\begin{theorem}
    \TN{} is \WOne-hard on bipartite graphs and on split graphs when they are parameterized by the solution size.
\end{theorem}

Next, we show that \GN{} and \ZN{} as well as their connected variants are \WOne-hard.

\begin{theorem}\label{thm:gn}
    \CONC{\GN} and \CONC{\ZN} are \WOne-hard on co-bipartite graphs when parameterized by the solution size.
\end{theorem}

\begin{proof}
    Let $(G,k)$ be an instance of \BIP{} with $V(G) = A \dot\cup B$. We construct the graph $G'$ by making the sets $A$ and $B$ cliques. We claim that $G$ has a total dominating sequence of $A$ of length $k$ if and only if $G'$ has a (connected) dominating sequence ($Z$-sequence) of length~$k$. 
    The vertices in a total dominating sequence of $A$ in $G$ do not footprint any vertex of $A$ but only vertices of~$B$. Therefore, every total dominating sequence of $A$ in $G$ is a connected dominating sequence and a connected $Z$-sequence of $G'$.

    Now assume that $G'$ has a dominating sequence ($Z$-sequence) $\sigma = (v_1, \dots, v_k)$. First consider the case that $v_1$ is a vertex of $A$.  By \cref{assumption}, $G$ does not contain isolated vertices. Thus, $v_1$ footprints some vertex in $B$. All other elements of $\sigma$ must footprint some vertex of $B$ as $A \subseteq N[v_1]$. Thus, we are done if $\sigma$ only contains elements of $A$. 
    
    So assume that there is a vertex $v_i$ in $\sigma$ that lies in set $B$. It holds that $N[v_1] \cup N[v_i] = V(G)$. Hence, $v_i$ must be the last vertex of $\sigma$, i.e., $v_i = v_k$. Since $A \subseteq N[v_1]$, vertex $v_k$ must footprint some vertex $b \in B$. By \cref{assumption}, vertex $b$ was not isolated in $G$ and, thus, there is a vertex $a \in A$ that is adjacent to $b$. Note that $a$ cannot be part of $\sigma$ since otherwise $v_k$ would not footprint $b$. We replace $v_k$ by $a$ in $\sigma$ and get a dominating sequence ($Z$-sequence) of $G$ with $k$ elements that only contains vertices of $A$. Then the argumentation above implies that $\sigma$ also forms a total dominating sequence of $A$ in $G$. 

    If $\sigma$ starts in $B$, then we can use the same argumentation and get a total dominating sequence of $B$ of length $k$. Due to \cref{lemma:A-B}, there is also such a sequence of~$A$.
\end{proof}

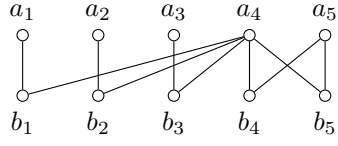
\begin{figure}
    \centering
    \begin{tikzpicture}[yscale=0.8]
    \node[svertex, label=90:$a_1$] (1-a) at (0,1) {};
    \node[svertex, label=90:$a_2$] (2-a) at (1,1) {};
    \node[svertex, label=90:$a_3$] (3-a) at (2,1) {};
    \node[svertex, label=90:$a_4$] (4-a) at (3,1) {};
    \node[svertex, label=90:$a_5$] (5-a) at (4,1) {};
    \node[svertex, label=-90:$b_1$] (1-b) at (0,0) {};
    \node[svertex, label=-90:$b_2$] (2-b) at (1,0) {};
    \node[svertex, label=-90:$b_3$] (3-b) at (2,0) {};
    \node[svertex, label=-90:$b_4$] (4-b) at (3,0) {};
    \node[svertex, label=-90:$b_5$] (5-b) at (4,0) {};

    \foreach \x in {1,...,5} 
    {
    \draw (\x-a) -- (\x-b);
    }

    \foreach \x in {1,...,5} 
    {
    \draw (4-a) -- (\x-b);
    }

    \draw (5-a) -- (4-b);
\end{tikzpicture}
    \caption{There is no ordering of the vertices $\{a_1, \dots, a_5\}$ that is a total dominating sequence. However, if we make $A$ and $B$ cliques, then the sequence $(a_1, a_2, a_3, b_4, b_5)$ is a total dominating sequence and an $L$-sequence.}
    \label{fig:gn-tn}
\end{figure}

This proof does not work for \TN{} and \LN{}. An example is given in \cref{fig:gn-tn}. However, we can slightly modify the construction to also prove \WOne-hardness of these problems on co-bipartite graphs.

\begin{theorem}\label{thm:tn}
    \CONC{\LN} and \CONC{\TN} are \WOne-hard on co-bipartite graphs when they are parameterized by the solution size.
\end{theorem}

\begin{proof}
    We use a similar construction as for \cref{thm:gn}. Let $(G,k)$ be an instance of \BIP{} with $V(G) = A \dot\cup B$. Without loss of generality, we may assume that $k > 0$. To construct the graph $G'$, we add two vertices $a_1$ and $a_2$ to $A$ and get the set $A'$. We also add two vertices $b_1$ and $b_2$ to $B$ and get the set $B'$. Finally, we make $A'$ and $B'$ cliques. We claim that $G$ has a total dominating sequence of $A$ with $k$ elements if and only if $G'$ has a (connected) total dominating sequence ($L$-sequence) with $\ell = k + 4$ elements.

    Assume first that $G$ has a total dominating sequence $\sigma = (v_1, \dots, v_k)$ of~$A$. Let $w$ be a vertex in $B$ that is footprinted by some vertex in $\sigma$. Then, we claim that $\sigma' = (a_1, a_2, v_1, \dots, v_n, b_1, w)$ is a connected total dominating sequence ($L$-sequence) of $G'$. Vertex $a_1$ footprints $a_2$ while $a_2$ footprints $a_1$. Similarly, $b_1$ footprints $b_2$ and $w$ footprints $b_1$. As none of the vertices $a_1$ and $a_2$ is adjacent to some vertex in $B'$, vertex $v_i$ footprints a vertex of $B$ in $\sigma$ if and only if it footprints the same vertex in $\sigma'$. Furthermore, $\sigma$ induces a connected graph since $w$ has some neighbor in $A$ that is part of $\sigma$.

    So assume now that $G'$ has a total dominating sequence ($L$-sequence) $\sigma = (v_1, \dots, v_\ell)$. First consider the case that $v_1 \in A'$. We claim that $v_2$ is also in $A'$. Assume for contradiction that this is not the case, i.e., $v_2 \in B'$. Then $N(v_1) \cup N(v_2) \supseteq V(G) \setminus \{v_1, v_2\}$. Hence, there are at most two vertices left that can be footprinted, implying that $\ell \leq 4$. This is a contradiction since $\ell = k + 4 > 4$. 
    
    Thus, $v_2 \in A'$ and $A' \subseteq N(v_1) \cup N(v_2)$. If $\sigma$ contains a vertex $v_i \in B'$, then $N(v_1) \cup N(v_2) \cup N(v_i) \supseteq A' \cup B' \setminus \{v_i\}$. Hence, there is at most one further vertex of $B'$ in~$\sigma$. This implies that there are at least $\ell - 2 = k + 2$ vertices of $A'$ in~$\sigma$. It follows from the argumentation above that only $v_1$ and $v_2$ might footprint some vertex of $A'$. Thus, the $k$ rightmost vertices $\sigma' = (u_1, \dots, u_k)$ of $A'$ in $\sigma$ must footprint some vertex of $B'$. Since a vertex $v \in A'$ can only footprint a vertex $w \in B'$ if $v \in A$ and $w \in B$, it holds that the $k$ rightmost vertices of $A'$ in $\sigma$ form a total dominating sequence of $A$ in $G$.
    
    If $\sigma$ starts in $B'$, then we can use the same argumentation and get a total dominating sequence of $B$ of length $k$. Due to \cref{lemma:A-B}, there is also such a sequence of~$A$.
\end{proof}

All reductions used to prove \cref{thm:gn,thm:tn} are polynomial-time reductions. Hence, they also imply \NP-completeness of the problems on co-bipartite graphs.

\begin{corollary}\label{corol:tn-cb}
    \CONC{\GN}, \CONC{\TN}, \CONC{\LN}, and \CONC{\ZN} are \NP-complete on co-bipartite graphs.
\end{corollary}

Note that the \NP-completeness of \GN{} on co-bipartite graphs has already been proven in~\cite{bresar2023computation}.  A simple reduction shows that the hypergraph version of Grundy domination is also \WOne-hard.

\begin{lemma}\label{thm:w1-hyper}
    \HYPER{} is \WOne-hard when it is parameterized by the solution size.
\end{lemma}

\begin{proof}
    We reduce from \BIP{}. Let $(G,k)$ be an instance with $V(G) = A \dot\cup B$. Let $\H = (X, \E)$ be the hypergraph with $X = B$ and $\E$ contains for every $v \in A$ the set of neighbors of $v$ in $G$. It is easy to see that a sequence $(v_1, \dots, v_k)$ is a total dominating sequence of $A$ if and only if $(N(v_1), \dots, N(v_k))$ is a covering sequence of~$\H$.
\end{proof}

Finally, we extend all \WOne-hardness results to \WOne-completeness. 

\begin{theorem}\label{thm:w3}
    The following problems are $\WOne$-complete when they are parameterized by the solution size:
    \begin{itemize}
        \item \CONC{\GN}, 
        \item \CONC{\TN}, 
        \item \CONC{\LN}, 
        \item \CONC{\ZN},
        \item \BIP{},
        \item \HYPER{}.
    \end{itemize}
\end{theorem}

\begin{proof}
    We have shown in \cref{thm:bip,thm:gn,thm:tn} as well as in \cref{thm:w1-hyper} that all these problems are \WOne-hard. It remains to show that the problems are also in \WOne{}. We first show that \LN{} is in $\WOne$. We then present a simple adaption of the proof to the other unconnected problems. Afterwards, we present a more sophisticated adaption for the connected problem variants. Finally, we use known \FPT{} reductions to prove \WOne-completeness for \BIP{} and \HYPER{}.
    
    \subparagraph{\LN{}} Let $(G,k)$ be an instance of \LN{}. In the following, we present a Boolean circuit of weft~$1$ and depth~$\leq 4$ that has a satisfying assignment of weight exactly $2k$ if and only if $G$ has an $L$-sequence of length $k$. For every $v \in V(G)$ and every $i \in [k]$, we introduce the variables $x(v,i)$ and $y(v,i)$. The variable $x(v,i)$ should be \true{} if and only if vertex $v$ is the $i$-th element of the $L$-sequence. The variable $y(v,i)$ should be \true{} if vertex $v$ is footprinted by the $i$-th element of the $L$-sequence. Note that we will only mark one vertex that is footprinted by the $i$-th vertex although there might be more than one such vertex. Further note that in $L$-sequences a vertex may be footprinted twice, once by itself and a second time by another vertex.
    
    First, we give Boolean formulas which ensure that the assignment of the $x$-variables encodes a sequence of $k$ pairwise different vertices. To this end, we define the following subformulas:
    \begin{align*}
    X^i_{vw} &:= \lnot{x(v,i)} \lor \lnot{x(w,i)}
    &X^{ij}_v &= \lnot{x(v,i)} \lor \lnot{x(v,j)}
    \end{align*}
    
    Formula $X^i_{vw}$ ensures that not both $v$ and $w$ can have the $i$-th position in the sequence. $X^{ij}_v$ ensures that vertex $v$ cannot be both at position $i$ and at position $j$. Using these formulas, we can give the following formula. 
    \begin{align*}
    X &:= \big(\bigwedge_{i \in [k]} \bigwedge_{\substack{v,w \in V(G) \\ v \neq w}} X^i_{vw}\big) \land \big(\bigwedge_{v \in V(G)} \bigwedge_{\substack{i,j \in [k] \\ i \neq j}} X^{ij}_v\big)
    \end{align*}

    Formula $X$ ensures that there is at most one vertex at every position $i \in [k]$ and that every vertex $v \in V(G)$ is placed at maximal one position in the sequence. Therefore, the $x$-variables represent a sequence of at most $k$ distinct vertices.

    We define a similar formula $Y$ for the $y$-variables as follows. 

    \begin{align*}
    Y &:= \bigwedge_{i \in [k]} \bigwedge_{\substack{v,w \in V(G) \\ v \neq w}} \big(\lnot{y(v,i)} \lor \lnot{y(w,i)}\big)
    \end{align*}

    This formula ensures that for every $i$ there is at most one $y$-variable set to \true{}. In difference to formula $X$, we do not force a vertex to have at most one true $y$-variable since a vertex may be footprinted twice.
    We will later force that both $X$ and $Y$ are satisfied. Since the weight of the satisfying assignment is $2k$, this ensures that there are exactly $k$ $x$-variables and $k$ $y$-variables set to \true{}.
    
    Next, we want to ensure that the $x$- and the $y$-variables match each other. To this end, we consider the following formulas:
    \begin{align*}
    A^{ij}_{wv} &:= \lnot y(w,i) \lor \lnot x(v,j) & & \forall i,j \in [k] \text{ with } j < i, \forall v,w \in V(G) \text{ with } w \in N(v) \\
    B^i_{wv} &:= \lnot y(w,i) \lor \lnot x(v,i) & & \forall i \in [k], \forall v,w \in V(G) \text{ with } w \notin N[v]
    \end{align*}
    Let $A$ be the conjunction of all formulas $A^{ij}_{wv}$ and $B$ be the conjunction of all formulas $B^{i}_{wv}$. For a fixed $w$ and $i$, $A$ ensures that if vertex $w$ is footprinted at position $i$, then it is not in the open neighborhood of any vertex that is placed at some position before~$i$. Formula $B$ ensures that if vertex $w$ is footprinted at position $i$, then the vertex placed at position $i$ contains $w$ in its closed neighborhood. 
    
    We define the final formula as $F := X \wedge Y \wedge A \wedge B$. It is easy to check that the circuit of $F$ has weft~$1$ and depth~$\leq 4$. By the above arguments, the formula has a satisfying assignment with weight $2k$ if and only if there is an $L$-sequence of length $k$ in $G$.

    \subparagraph{Other Unconnected Problems} The circuit works analogously for the other problems. We just have to change the usage of open and closed neighborhoods in the definitions of $A^{ij}_{wv}$ and $B^i_{wv}$. For \TN{}, we keep the definition of $A^{ij}_{wv}$. For \GN{} and \ZN{}, we use $N[v]$ instead of $N(v)$. Analogously, we keep the definition of $B^i_{wv}$ for \GN{}; for \TN{} and \ZN{} we replace $N[v]$ by $N(v)$ in that definition.

    \subparagraph{Connected Problems} To ensure that the chosen vertices of the sequence induce a connected subgraph, we have to force that for every non-trivial subset of them, there is an edge from a vertex within the subset to the vertices of the sequence that are not in the subset. We introduce further variables for that. For every non-trivial subset $\emptyset \neq S \subsetneq [k]$ and every pair $i,j$ with $i \in S$ and $j \in [k] \setminus S$, we introduce the variable $z(S,i,j)$. This variable should be \true{} if the $i$-th vertex in the sequence is adjacent to the $j$-th vertex in the sequence. Similar as for the $y$-variables, we only use one certifying pair although there might be several edges between $S$ and $[k] \setminus S$. Therefore, we consider the following formulas.
    \begin{align*}
        Z^{iji'j'}_S &:= \lnot z(S,i,j) \lor \lnot z(S,i',j') &\forall \emptyset \neq S \subsetneq [k], i, i' \in S, j, j' \in [k] \setminus S \\ &&\text{ with } (i,j) \neq (i',j')
    \end{align*}
    Let $Z$ be the conjunction of all formulas $Z^{iji'j'}_S$. Formula $Z$ ensures that for every subset $S$ there is at most one variable set to \true{}. Thus, there are at most $2^k - 2$ $z$-variables set to \true{}. It remains to ensure that the assignment of the $z$-variables matches the adjacencies. To this end, we consider the following formulas.
    \begin{align*}
        \C^{ij}_{Svw} &:= \lnot z(S,i,j) \lor \lnot x(v,i) \lor \lnot x(w,j) & \forall \emptyset \neq S \subsetneq [k], i \in S, j \in [k] \setminus S, \\ &&v \in V(G), w \notin N(v) 
    \end{align*}
    Let $C$ be the conjunction of all these formulas. Formula $C$ ensures that if the $i$-th and the $j$-th vertex are chosen to be adjacent, then they also are adjacent.
    
    We now replace the final formula $F$ by $F \land Z \land C$. Note that the new formula has also weft~1 and its depth increases by~1. We now ask whether this formula has a satisfying assignment with weight $2k + 2^{k} - 2$. If this is the case, then by the above mentioned upper bounds on \true{} variables, there are exactly $k$ $x$-variables, exactly $k$ $y$-variables and exactly $2^k-2$ $z$-variables that are set to \true{}. Therefore, there is a connected dominating sequence of the particular type in $G$. Conversely, such a sequence can directly be transformed into a satisfying assignment of weight $2k + 2^{k} - 2$.

    \subparagraph{One-Sided and Hypergraph Problem} Due to \cref{lemma:A-B}, \TN{} and \BIP{} are \FPT-equivalent. Furthermore, the polynomial-time reduction from \HYPER{} to \TN{} given in \cite{bresar2016total} is an \FPT{} reduction. Thus, both problems are also in~\WOne.
\end{proof}

\section{Zero Forcing Problems}\label{sec:zero}

\subsection{Definitions}

Zero forcing problems are based on a domination process where one is allowed to color some vertices of the graph initially blue while all other vertices are colored white. Then one can use certain color change rules to color the white vertices blue. There have been given a wide range of such color change rule schemes in the literature (see, e.g., \cite{barioli2013parameters,hogben2022inverse,lin2019zero}). We break some of these rule schemes down into the following three basic rules (see \cref{fig:rules} for an illustration).

\begin{description}
    \item[Z-rule] If $v$ is a blue vertex and $v$ has exactly one white neighbor $w$, then change the color of $w$ to blue. We write $v \z w$.
    \item[T-rule] If $v$ is a white vertex and $v$ has exactly one white neighbor $w$, then change the color of $w$ to blue. We write $v \t w$.
    \item[D-rule] If $v$ is a white vertex and every neighbor of $v$ is blue, then change the color of $v$ to blue. We write $v \d v$.
\end{description}

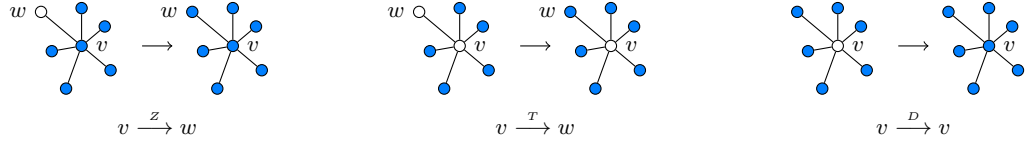
\begin{figure}
    \centering
    \begin{tikzpicture}
\footnotesize
    \begin{scope}
    \node[svertex, fill=myblue, label=0:$v$] (1) at (0,0) {};
    \node[svertex, fill=myblue] (2) at (250:0.6cm) {};3
    \node[svertex, fill=myblue] (3) at (190:0.4cm) {};
    \node[svertex, label=180:$w$] (4) at (140:0.7cm) {};
    \node[svertex, fill=myblue] (5) at (90:0.5cm) {};
    \node[svertex, fill=myblue] (6) at (40:0.4cm) {};
    \node[svertex, fill=myblue] (7) at (-40:0.5cm) {};

    \draw (1) -- (2);
    \draw (1) -- (3);
    \draw (1) -- (4);
    \draw (1) -- (5);
    \draw (1) -- (6);
    \draw (1) -- (7);
    \end{scope}

    \node at (1,-1) {$v \z w$};
    \draw[->] (0.8,0) -- (1.2,0);
    
    \begin{scope}[xshift=2cm]
    \node[svertex, fill=myblue, label=0:$v$] (1) at (0,0) {};
    \node[svertex, fill=myblue] (2) at (250:0.6cm) {};3
    \node[svertex, fill=myblue] (3) at (190:0.4cm) {};
    \node[svertex, fill=myblue, label=180:$w$] (4) at (140:0.7cm) {};
    \node[svertex, fill=myblue] (5) at (90:0.5cm) {};
    \node[svertex, fill=myblue] (6) at (40:0.4cm) {};
    \node[svertex, fill=myblue] (7) at (-40:0.5cm) {};

    \draw (1) -- (2);
    \draw (1) -- (3);
    \draw (1) -- (4);
    \draw (1) -- (5);
    \draw (1) -- (6);
    \draw (1) -- (7);
    \end{scope}

    \begin{scope}[xshift=5cm]
    \node[svertex, label=0:$v$] (1) at (0,0) {};
    \node[svertex, fill=myblue] (2) at (250:0.6cm) {};3
    \node[svertex, fill=myblue] (3) at (190:0.4cm) {};
    \node[svertex, label=180:$w$] (4) at (140:0.7cm) {};
    \node[svertex, fill=myblue] (5) at (90:0.5cm) {};
    \node[svertex, fill=myblue] (6) at (40:0.4cm) {};
    \node[svertex, fill=myblue] (7) at (-40:0.5cm) {};

    \draw (1) -- (2);
    \draw (1) -- (3);
    \draw (1) -- (4);
    \draw (1) -- (5);
    \draw (1) -- (6);
    \draw (1) -- (7);
    \end{scope}

    \node at (6,-1) {$v \t w$};
    \draw[->] (5.8,0) -- (6.2,0);
    
    \begin{scope}[xshift=7cm]
    \node[svertex, label=0:$v$] (1) at (0,0) {};
    \node[svertex, fill=myblue] (2) at (250:0.6cm) {};3
    \node[svertex, fill=myblue] (3) at (190:0.4cm) {};
    \node[svertex, fill=myblue, label=180:$w$] (4) at (140:0.7cm) {};
    \node[svertex, fill=myblue] (5) at (90:0.5cm) {};
    \node[svertex, fill=myblue] (6) at (40:0.4cm) {};
    \node[svertex, fill=myblue] (7) at (-40:0.5cm) {};

    \draw (1) -- (2);
    \draw (1) -- (3);
    \draw (1) -- (4);
    \draw (1) -- (5);
    \draw (1) -- (6);
    \draw (1) -- (7);
    \end{scope}

    \begin{scope}[xshift=10cm]
    \node[svertex, label=0:$v$] (1) at (0,0) {};
    \node[svertex, fill=myblue] (2) at (250:0.6cm) {};3
    \node[svertex, fill=myblue] (3) at (190:0.4cm) {};
    \node[svertex, fill=myblue] (4) at (140:0.7cm) {};
    \node[svertex, fill=myblue] (5) at (90:0.5cm) {};
    \node[svertex, fill=myblue] (6) at (40:0.4cm) {};
    \node[svertex, fill=myblue] (7) at (-40:0.5cm) {};

    \draw (1) -- (2);
    \draw (1) -- (3);
    \draw (1) -- (4);
    \draw (1) -- (5);
    \draw (1) -- (6);
    \draw (1) -- (7);
    \end{scope}

    \node at (11,-1) {$v \d v$};
    \draw[->] (10.8,0) -- (11.2,0);
    
    \begin{scope}[xshift=12cm]
    \node[svertex, fill=myblue, label=0:$v$] (1) at (0,0) {};
    \node[svertex, fill=myblue] (2) at (250:0.6cm) {};3
    \node[svertex, fill=myblue] (3) at (190:0.4cm) {};
    \node[svertex, fill=myblue] (4) at (140:0.7cm) {};
    \node[svertex, fill=myblue] (5) at (90:0.5cm) {};
    \node[svertex, fill=myblue] (6) at (40:0.4cm) {};
    \node[svertex, fill=myblue] (7) at (-40:0.5cm) {};

    \draw (1) -- (2);
    \draw (1) -- (3);
    \draw (1) -- (4);
    \draw (1) -- (5);
    \draw (1) -- (6);
    \draw (1) -- (7);
    \end{scope}
\end{tikzpicture}
    \caption{The three color-change rules considered here.}
    \label{fig:rules}
\end{figure}

We will also use the constants $Z$, $T$, and $D$ to refer to the respective rule. As mentioned in the introduction, the $Z$-rule was considered for many different problems. In particular, it is the defining rule of \emph{power dominating sets}~\cite{haynes2002domination} and \emph{zero forcing sets}~\cite{aim2008zero}. A combination of $Z$-rule and $T$-rule was introduced in~\cite{ima2010minimum} for \emph{skew zero forcing sets}. The $D$-rule was used in~\cite{barioli2013parameters} in combination with the $Z$-rule to define the \emph{loop zero forcing sets}. 

Given a graph $G$ and some non-empty rule set $\R \subseteq \ZTD$, we define a set $S \subseteq V(G)$ to be an \emph{$\R$-forcing set} of $G$ if the following procedure is able to color all vertices of $G$ blue:
\begin{enumerate}
    \item Color the vertices of $S$ blue and all vertices of $V(G) \setminus S$ white.
    \item Iteratively apply some rule of $\R$ to the vertices of $G$.
\end{enumerate}

If the vertices in $S$ induce a connected subgraph, then we say that $S$ is a \emph{connected \R-forcing set}~\cite{brimkov2016characterizations}. If the vertices in $S$ induce a graph without isolated vertices, then $S$ is a \emph{total \R-forcing set}~\cite{davila2015bounding}. These definitions give rise to the following parameterized problem.
\begin{center}
\fbox{\parbox{0.98\textwidth}{\noindent
{\FORCE{\R}} for some non-empty rule set $\R \subseteq \ZTD$\\[.8ex]
\begin{tabular*}{.93\textwidth}{rl}
{\em Input:} & A graph $G$, $k \in \N$.\\
{\em Parameter:} & $k$ \\
{\em Question:} & Is there an $\R$-forcing set of $G$ of size at most $k$?\end{tabular*}
}}
\end{center}

We will use the name \CON{\FORCE{\R}} and \TOT{\FORCE{\R}} for the variants where we look for a connected or total \R-forcing sets, respectively. Note that for better readability, we will leave out the set parentheses of the rule set~$\R$ when considering concrete variants of \FORCE{\R}. The problem $\FORCE{Z}$ has been considered extensively before  under the name \ZFS{}. As mentioned before, there is a strong relation between these forcing set problems and the parametric duals of the Grundy domination problems. The dual version of \GN{} is defined as follows.

\begin{center}
\fbox{\parbox{0.98\textwidth}{\noindent
{\DUAL{\GN}}\\[.8ex]
\begin{tabular*}{.93\textwidth}{rl}
{\em Input:} & A graph $G$ with $n$ vertices, $k \in \N$.\\
{\em Parameter:} & $k$ \\
{\em Question:} & Is there a dominating sequence of $G$ of length at least $n -k$?
\end{tabular*}
}}
\end{center}

We define the dual versions of  \TN{}, \ZN{}, and \LN{} equivalently. Using this terminology, we can state a result of Lin~\cite{lin2019zero} as follows.

\begin{theorem}[Lin~{\cite[Lemma 2.1]{lin2019zero}}]\label{thm:lin}
The following equivalences between problems hold, i.e., an instance is a yes-instance for one of the problems if and only if it is a yes-instance for the other problem.
  \begin{itemize}
      \item \FORCE{Z} $\equiv$ \DUAL{\ZN}
      \item \FORCE{Z,D} $\equiv$ \DUAL{\GN}
      \item \FORCE{Z,T} $\equiv$ \DUAL{\TN}
      \item \FORCE{Z,T,D} $\equiv$ \DUAL{\LN}
  \end{itemize}    
\end{theorem}

Note that the first equivalence has been proven already by Bre\v{s}ar et al.~\cite{bresar2017grundy}. Further note that the proofs of Bre\v{s}ar et al.~\cite{bresar2017grundy} and Lin~\cite{lin2019zero} present more structural insights into the solutions of these two problem variants. They show that the complement set of the respective version of a dominating sequence is an $\R$-forcing set and vice versa. Even stronger, the ordering of the sequence directly implies the ordering of the color change rule applications and vice versa.

The following lemma about the $T$-rule will help us throughout the section.

\begin{lemma}\label{lemma:t-independent}
    Let $\R \subseteq \{Z,T,D\}$ with $T \in \R$. Let $S$ be an $\R$-forcing set of a graph $G$ and let $(R_1, \dots, R_\ell)$ be the sequence of $\R$-rules applied to color $G$ blue. Let $A$ be those vertices of $G$ to which a $T$-rule is applied, i.e., vertex $v$ is in $A$ if and only if there is some rule $R_i = v \t w$ for some $w \in V(G)$. Then, $A$ induces an independent set in $G$. In particular, there is no rule $R_i = v \t w$ with $w \in A$.
\end{lemma}

\begin{proof}
    Assume for contradiction that there is an edge $uv \in E(G)$ with $u,v \in A$. Without loss of generality, we may assume that $u$ applies its $T$-rule before $v$, i.e., $v$ is still white when $u$ applies its $T$-rule. This implies that $v$ must be the neighbor of $u$ that is colored blue by the $T$-rule. However, this contradicts the fact that some $T$-rule will be applied to $v$ later.
\end{proof}

\subsection{A Short Intermezzo: Local L-Sequences}\label{sec:lln}

There are seven non-empty subsets of $\ZTD$ of which only four are mentioned in \cref{thm:lin}. In this subsection, we will consider the other three. The results are not needed to understand the following sections. So the impatient reader can skip this subsection and can directly jump to \cref{sec:treewidth}.

Two of the remaining subsets of $\ZTD$ are rather uninteresting as the following observation shows.

\begin{observation}\label{obs:vc}
The following statements hold.
\begin{itemize}
\item A set $S \subseteq V(G)$ is a $\{D\}$-forcing set of $G$ if and only if $S$ is a vertex cover of $G$.
\item A set $S \subseteq V(G)$ is a $\{T\}$-forcing set of $G$ if and only if $S = V(G)$.
\end{itemize}
\end{observation}

\begin{proof}
    Let $S$ be a $\{D\}$-forcing set of $G$. If there is an edge $xy$ with $x,y \in V(G) \setminus S$, then the $D$-rule can never be applied to $x$ or $y$. So $S$ is a vertex cover of $G$. Conversely, every vertex cover is a $\{D\}$-forcing set since we can apply the $D$-rule to every white vertex.

    Let $S$ be a $\{T\}$-forcing set of $G$. Assume for contradiction that $w \in V(G) \setminus S$. Then $w$ becomes blue because the $T$-rule $v \t w$ is applied for some white $v \in V(G) \setminus S$ with $v \neq w$. However, $v$ has also be colored blue by a $T$-rule; a contradiction to \cref{lemma:t-independent}.
\end{proof}

The only remaining subset of $\ZTD$ is $\TD$. The problem $\FORCE{T,D}$ can also be related to some Grundy domination problem that -- to the best of our knowledge -- has not been considered so far. This problem concerns so-called \emph{local $L$-sequences}.

\begin{definition}
    A sequence $(v_1,\dots,v_k)$ of vertices of $G$ is a \emph{local $L$-sequence} of $G$ if for all $i \in [k]$ it holds that $(N[v_i] \cap \{v_1, \dots, v_i\}) \setminus \bigcup_{j=1}^{i-1} N(v_j)$ is not empty.
\end{definition}

So, local $L$-sequences are defined similarly to $L$-sequences, except that a vertex is only allowed to footprint a vertex that is to the left of it in the sequence. Note that -- in difference to the other forcing sets and dominating sequences -- $\TD$-forcing sets and local $L$-sequences are not monotone. This means that a superset of a $\TD$-forcing set is not necessarily also a $\TD$-forcing and a subsequence of a local $L$-sequence is not necessarily also a local $L$-sequence.

The following theorem presents the relation between local $L$-sequences and $\TD$-forcing sets. 

\begin{theorem}\label{thm:local-l}
    Let $G$ be a graph with $n$ vertices. There is a local $L$-sequence of length $k$ in $G$ if and only if there is a $\TD$-forcing set of size $n-k$ in $G$.
\end{theorem}

\begin{proof}
    The proof follows the same arguments as the proof of \cref{thm:lin} given by Lin~\cite{lin2019zero}. First assume that there is a local $L$-sequence $(v_1, \dots, v_k)$ in $G$. For every $i \in [k]$, let $u_i$ be a vertex contained in $(N[v_i] \cap \{v_1, \dots, v_i\}) \setminus \bigcup_{j=1}^{i-1} N(v_j)$. We claim that $V(G) \setminus \{v_1, \dots, v_k\}$ is a $\TD$-forcing set. Let $(R_k, \dots, R_1)$ be a sequence of color change rules where $R_i$ is $v_{i} \d v_{i}$ if $v_{i} = u_{i}$ or, otherwise, $R_i$ is $u_i \t v_i$. 

    We claim that the following invariant holds during the application of this rule sequence.  
     \begin{description}
         \item[Invariant:] Before rule $R_i$ is applied, all vertices in $\{v_1, \dots, v_i\}$ are white and all other vertices are blue.
     \end{description} 
    The invariant is true before the first rule $R_k$ is applied since the vertices $v_1, \dots, v_k$ are not part of the \TD-forcing set, but all other vertices are part of that forcing set. Now let $i \leq k$. If $R_i$ is $v_i \d v_i$, then, by definition, $u_i = v_i$. Therefore, $v_i$ is not in the neighborhood of any vertex in $\{v_1, \dots, v_{i-1}\}$. Due to the invariant, all vertices $v_j$ with $j > i$ are already blue, implying that all neighbors of $v_i$ are blue. Thus, we can apply $v_i \d v_i$. If $R_i$ is $u_i \t v_i$, then $u_i$ is not in the neighborhood of any vertex in $\{v_1, \dots, v_{i-1}\}$. Due to the invariant, $v_i$ is the only white neighbor of $u_i$. Since $u_i \in \{v_1, \dots, v_{i-1}\}$, it also holds that $u_i$ is white. Hence, we can apply $u_i \t v_i$. Note that in both cases only $v_i$ becomes blue. Thus, all vertices in $\{v_1, \dots, v_{i-1}\}$ are white and all other vertices are blue, i.e., the invariant holds before the next rule $R_{i-1}$ is applied.

    For the reverse direction, let $S$ be a $\TD$-forcing set of $G$ and let $k = |S|$. Let $(R_1, \dots, R_\ell)$ be the sequence of $T$-rules and $D$-rules used to color all vertices of $V(G) \setminus S$ blue. If $R_i$ is a $D$-rule, then we call the corresponding vertex $w_i$, i.e., $R_i = w_i \d w_i$. Similarly, if $R_i$ is a $T$-rule, then $R_i = v_i \t w_i$. Note that every vertex that is not in $S$ appears exactly once as $w_i$. This implies that $\ell = n - k$. Note that $w_i$ can additionally appear also once as $v_h$ with $h < i$ if $R_i$ is a $D$-rule. We claim that $\sigma = (w_\ell, \dots, w_1)$ is a local $L$-sequence of $G$. Let $i \in [\ell]$ be arbitrary. If $R_i = w_i \d w_i$, then all neighbors of $w_i$ have been blue. 
    This implies that all neighbors of $w_i$ are either part of $S$ or are some vertex $w_h$ with $h < i$. Therefore, none of the vertices to the left of $w_i$ in $\sigma$ is a neighbor of $w_i$ and $w_i \in N[w_i] \cap \{w_i, \dots, w_\ell\} \setminus \bigcup_{j=i+1}^{\ell} N(w_j)$. If $R_i = v_i \t w_i$, then $v_i$ has been white. This implies that $v_i = w_h$ for some $h > i$. Furthermore, $w_i$ has been the only white neighbor of $w_h$. Thus, none of the vertices in $\{w_{i+1}, \dots, w_\ell\} \setminus \{w_h\}$ is adjacent to $w_h$ and $w_h \in N[w_i] \cap \{w_i, \dots, w_\ell\} \setminus \bigcup_{j=i+1}^{\ell} N(w_j)$.
\end{proof}

We define the (parameterized) problem \LLN{} where one has to determine whether there is a local $L$-sequence of length $\geq k$. Note that, due to \cref{thm:local-l}, \FORCE{T,D} is equivalent to \DUAL{\LLN}. 

In the following, we will settle the (parameterized) complexity of \LLN. We start with a lemma that relates local $L$-sequences to the \emph{independence number} $\alpha(G)$ of a graph $G$, i.e., the size of the largest independent set of $G$.

\begin{lemma}\label{lemma:local-independent}
    Let $G$ be a graph. Every local $L$-sequence of $G$ of length $\ell$ contains an independent set of size $\left\lceil \frac{\ell}{2}\right\rceil$. Therefore, if $k$ is the length of the longest local $L$-sequence of $G$, then $\alpha(G) \leq k \leq 2 \alpha(G)$.
\end{lemma}

\begin{proof}
    Every ordering $(v_1, \dots, v_\ell)$ of an independent set forms a local $L$-sequence since $v_i \in (N[v_i] \cap \{v_1, \dots, v_i\}) \setminus \bigcup_{j=1}^{i-1} N(v_j)$. Therefore, $\alpha(G) \leq k$.

    So let $\sigma = (v_1, \dots, v_\ell)$ be an arbitrary local $L$-sequence. We distinguish two different types of vertices in that sequence. We call a vertex $v_i$ an $A$-vertex if it is not adjacent to any vertex $v_j$ with $j < i$. Otherwise, we call the vertex $B$-vertex. 
    
    By definition, the $A$-vertices form an independent set of $G$. Thus, their number is bounded by $\alpha(G)$. Let $v_i$ be a $B$-vertex. Since it has a neighbor to the left in $\sigma$, $v_i \notin N[v_i] \setminus \bigcup_{j=1}^{i-1} N(v_j)$. Therefore, $v_i$ footprints some vertex that is to the left of it in $\sigma$. This vertex cannot be a $B$-vertex as these vertices have already been seen by other vertices. Thus, it is an $A$-vertex. Furthermore, every $A$-vertex can be footprinted at most once by a vertex different from itself. Therefore, there is an injective mapping from the $B$-vertices to the $A$-vertices and the number of $B$-vertices in $\sigma$ is bounded by the number of $A$-vertices. This implies that $\sigma$ contains an independent set of size at least $\left\lceil \frac{\ell}{2}\right\rceil$ and, thus, $\ell \leq 2\alpha(G)$.
\end{proof}

This directly implies an \XP{} algorithm for \CONC{\LLN} when parameterized by $\alpha(G)$.

\begin{corollary}\label{corol:lln-alpha}
    \CONC{\LLN} can be solved in $n^{\O(\alpha(G))}$ time given graph~$G$. In particular, it can be solved in polynomial time on co-bipartite graphs.
\end{corollary}

This contrasts with the other Grundy domination problems considered in \cref{sec:grundy} which are \NP-complete on co-bipartite graphs and, hence, on graphs of independence number~2 (see \cref{corol:tn-cb}). Nevertheless, we can prove that \CONC{\LLN} is \WOne-complete and \NP-complete.

\begin{theorem}\label{thm:lln}
    \CONC{\LLN} is \NP-complete and \WOne-complete when it is parameterized by the solution size plus the graph's independence~number.
\end{theorem}

\begin{proof} 
We first present the \WOne-hardness (and \NP-hardness) proof. Then we explain how we can adapt the proof of \cref{thm:w3} to also show \WOne-completeness.
\subparagraph{Hardness}
    We reduce from \MCP{}. Let $G$ be the input graph and $V^1, \dots, V^k$ be the partition of $V(G)$ into independent sets. We consider the complement $\overline{G}$ of $G$ and construct a graph $G'$ from $\overline{G}$ as follows. For every $i \in [k]$, we add a vertex $x^i$ that is adjacent to all vertices of $V^i$. Furthermore, we add two vertices $y^{k+1}$ and $x^{k+1}$ that are adjacent and form the set $V^{k+1}$. Let $X = \{x^i \mid i \in [k+1]\}$. We add edges such that the set $X \cup \{y^{k+1}\}$ forms a clique. Note that $\alpha(G') \leq k+1$. We claim that $G$ has a multicolored clique of size $k$ if and only if $G'$ has a local $L$-sequence of length $2k+2$.

    First assume that $G$ has a clique $S = \{v^1_{p_1}, \dots, v^k_{p_k}\}$, i.e., $S$ forms an independent set in~$G'$. Consider the sequence $\sigma = (v^1_{p_1}, \dots, v^k_{p_k}, y^{k+1}, x^{k+1}, x^1, \dots, x^k)$. Note that the vertices of this sequence induce a connected subgraph of $G'$. Since $S$ induces an independent set in $G'$, every vertex of $S$ footprints itself. The same holds for $y^{k+1}$. Note that all these vertices may be footprinted a second time since we consider $L$-sequences. Vertex $x^{k+1}$ footprints $y^{k+1}$. Every other vertex $x^i$ footprints $v^i_{p_i}$. Therefore, $\sigma$ is a connected local $L$-sequence of $G'$ of length $2k+2$.

    Now let $\sigma$ be an arbitrary local $L$-sequence of length $2k +2$. Due to \cref{lemma:local-independent}, $\sigma$ contains an independent set $S$ of size $k+1$. This set must contain exactly one vertex of every clique $V^i \cup \{x^i\}$ in $G'$ with $i \in [k+1]$. As both $x^{k+1}$ and $y^{k+1}$ are adjacent to all vertices in $X$, there cannot be any vertex $x^i$ in $S$ with $i \leq k$. Therefore, $S$ contains $k$ vertices of $G$ that form an independent set in $G'$ and, thus, a clique in $G$.

    \subparagraph{Completeness}
    The \NP-completeness trivially holds since the local $L$-sequence is a certificate that we can check in polynomial time. 
    For the \WOne-completeness, we adapt the proof of \cref{thm:w3} as follows. We keep the variables $x(v,i)$ that are \true{} if and only if vertex $v$ is the $i$-th vertex in the sequence. We also keep formula $X$ that ensures that every vertex $v$ is assigned to at most one position and to every position~$i$ there is assigned at most one vertex. 
    
    Since the vertices in a local $L$-sequence are only allowed to footprint vertices that are to left of them in the sequence, we have to change the definition of the $y$-variables. We use variables $y(j,i)$ with $j \leq i$. If this variable is set to \true{}, then the $i$-th vertex in the sequence footprints the $j$-th vertex in the sequence. We redefine the formula $Y$ as follows.
    \begin{align*}
        Y &:= \bigwedge_{i \in [k]} \bigwedge_{j < j' \leq i} \big(\lnot y(j,i) \lor \lnot y(j',i)\big)
    \end{align*}
    So $Y$ ensures that for every position $i$, there is at most one variable $y(j,i)$ set to \true{}. We use the following subformulas to guarantee that the assignment matches the adjacencies.
    \begin{align*}
        A^{ijj'}_{wv} &:= \lnot y(j,i) \lor \lnot x(w,j) \lor \lnot x(v,j') & \forall i,j,j' \in [k] \text{ with } j,j' \leq i, \\ && \forall v,w \in V(G) \text{ with } w \in N(v) \\
        B^{ij}_{wv} &:= \lnot y(j,i) \lor \lnot x(w,j) \lor \lnot x(v,i) & \forall i,j \in [k] \text{ with } j \leq i, \\ && \forall v,w \in V(G) \text{ with } w \notin N[v]
    \end{align*}
    The conjunction $A$ of all formulas $A^{ijj'}_{wv}$ ensures that if the vertex at position $i$ footprints the vertex at position $j$, then the vertex at position $j$ is not in the open neighborhood of any vertex to the left of position $i$. The conjunction $B$ of all formulas $B^{ij}_{wv}$ ensures that if the vertex at position $i$ footprints the vertex at position $j$, then one of them is contained in the closed neighborhood of the other. So the formula $F = X \land Y \land A \land B$ has a satisfying assignment with weight $2k$ if and only if there is a local $L$-sequence of $G$ of size $k$. For the connected variant, we can use the same approach as in the proof of \cref{thm:w3}.
\end{proof}

Due to \cref{thm:local-l}, this result implies \NP-completeness of $\FORCE{T,D}$.

\begin{corollary}
    $\FORCE{T,D}$ is \NP-complete.
\end{corollary}

Since the used reduction from \MCP{} increases the parameter value only linearly, \cref{thm:lower} implies the following.

\begin{corollary}
Assuming the Exponential Time Hypothesis, there is no $f(t) n^{o(t)}$ time algorithm for \CONC{\LLN} for any computable function~$f$ where $t = k + \alpha(G)$ and $n$ is the number of vertices of the graph.
\end{corollary}

\subsection{FPT Algorithms for Parameter Treewidth}\label{sec:treewidth}

Both Aazami~\cite{aazami2008hardness} and Bhyravarapu et al.~\cite{bhyravarapu2025parameterized} presented \FPT{} algorithms for \ZFS{} when parameterized by treewidth which are based on an approach of Guo et al.~\cite{guo2008improved} for \PDS{}. Instead of solving the problem directly, they solve an equivalent edge orientation problem. Here, we will present a more direct algorithm for all problems \FORCE{\R}. Nevertheless, this approach also uses some directed graphs as auxiliary substructures.

Our algorithm will use an extended version of \emph{nice} tree-decompositions. In these decompositions, the tree is considered to be a rooted binary tree. Furthermore, there are particular types of nodes. We will use a tree decomposition consisting of the following five node types.

\begin{description}
    \item[Leaf] Such a node is a leaf of $T$. Its bag is empty. 
    \item[Introduce Node] Such a node $t$ has exactly one child $t'$. There is a vertex $v \notin X_{t'}$ such that $X_{t} = X_{t'} \cup \{v\}$. We say that $v$ \emph{is introduced} at $t$. 
    \item[Forget Node] Such a node $t$ has exactly one child $t'$. There is a vertex $v \in X_{t'}$ such that $X_{t} = X_{t'} \setminus \{v\}$. We say that $v$ \emph{is forgotten} at $t$. 
    \item[Rule Node] Such a node $t$ has exactly one child $t'$ and it holds that $X_t = X_{t'}$. Furthermore, the parent of $t$ is a forget node. These tree nodes will be the place of our algorithm where we decide which particular rules are applied to the forgotten vertex.
    \item[Join Node] Such a node has exactly two children $t_1$ and $t_2$ and it holds that $X_t = X_{t_1} = X_{t_2}$.
\end{description}

Note that we also assume that the bag of the root node of the tree is empty, i.e., for every vertex of the graph there is exactly one forget node where this node is forgotten. For every $t \in V(T)$, we write $N_t(v)$ as short form of $N(v) \cap X_t$. Furthermore, for every $t \in V(T)$, we define $V^t := \{ v \mid v \in X_{t'} \text{ for some descendant $t'$ of $t$}\}$, where $t$ is considered a descendant of itself. As usual, our algorithm consists of a dynamic programming approach. For every node/bag of the tree decomposition, we will consider certain signatures that describe which decision we have made for the vertices within the bag. These signatures are six-tuples of the form $\Omega = (\Gamma, \Phi, b_\Gamma, b_\Phi, \D, \omega)$ whose entries are explained in the following.

\begin{description}
    \item[$\Gamma$-Type] For every vertex $v$, we have a value $\Gamma(v) \in \R \cup \{\bot\}$ describing how $v$ becomes blue. Here, $Z$, $T$, and $D$ stand for the respective rule, while $\bot$ describes that $v$ is part of the $\R$-forcing set.
    \item[$\Phi$-Type] For every vertex $v$, we have a value $\Phi(v) \in (\R \setminus \{D\}) \cup \{\bot\}$ describing whether $v$ colors some other vertex blue using some rule of $\R$. If this is not the case, then $\Phi(v) = \bot$. 
    \item[Function $b_\Gamma$] For every vertex $v$, we have a boolean value $b_\Gamma(v)$. This value is $\true$ if and only if $\Gamma(v) = \bot$ or we have already fixed some vertex $x \in V^t$ to be the vertex that colors $v$ blue by applying the rule $x \crule{\Gamma(v)} v$.
    \item[Function $b_\Phi$] For every vertex $v$, we have a boolean value $b_\Phi(v)$. This value is $\true$ if and only if $\Phi(v) = \bot$ or we have already fixed some vertex $x \in V^t$ to be the vertex that is colored blue by applying the rule $v \crule{\Phi(v)} x$.
    \item[Dependency Graph \D] The dependency graph $\mathfrak{D}$ is a directed graph. For every vertex $v \in X_t$, the digraph $\D$ contains an \emph{event node} $\gamma_v$ and, if $\Phi(v) \neq \bot$, it also contains an event node $\phi_v$. The vertex $\gamma_v$ represents the event when $v$ is made blue. The vertex $\phi_v$ represents the event when the $Z$-rule or $T$-rule is applied to $v$. An arc from one of these event nodes $x$ to another event node $y$ implies that the event represented by $x$ must happen before the event represented by $y$.
    \item[Weight $\omega$] Every signature is associated with a value $\omega \in \{0, \dots, n\} \cup \infty$ describing how many vertices in $V^t \setminus X_t$ have been assigned the $\Gamma$-type $\bot$, or, in case that $\omega = \infty$, describing that $\D$ is not acyclic. We call signatures with $\omega < \infty$ \emph{valid}.
\end{description}

The main ideas of the algorithm are the following:

\begin{itemize}
    \item When a vertex $v$ is introduced in its introduce node, then we fix the $\Gamma$- and $\Phi$-type of $v$, i.e., we specify which rule type should make $v$ blue and whether some rule is applied to $v$ to make some other vertex blue. Note that there are three combinations of $(\Gamma(v), \Phi(v))$ that do not have to be considered: $(\bot, T)$, since a vertex that is part of $S$ cannot apply a $T$-rule, $(D,Z)$, since a vertex colored blue with a $D$-rule has no white neighbor anymore, and $(T,T)$, due to \cref{lemma:t-independent}.
    
    For both decisions, we do not yet specify the other vertex of the corresponding rule as this vertex might not be part of the bag. Nevertheless, fixing the type of rule already implies some arcs for the dependency graph $\D$.
    \item We lazily fix the vertices of the rules in the rule nodes. This means, we only specify the other vertex $w$ of a rule concerning $v$ if either $v$ or $w$ is forgotten in the next bag. Again, fixing these vertices implies several arcs of the dependency graph~$\D$. These arcs also concern vertices that are neither $v$ nor $w$. For example, if we want to apply rule $v \z w$, then all other neighbors of $v$ in the bag must be blue before the $Z$-rule is applied to $v$. Conversely, assume $x$ is a neighbor of $v$ and it is fixed that some $T$-rule is applied to $x$ that colors some vertex different from $v$ blue. Then $v$ must become blue before the $T$-rule is applied to $x$. 
\end{itemize}

In the following, we will describe the detailed procedures for every possible type of tree node~$t$. In this descriptions, we will say that we \emph{bypass $v$ in $\D$} if we consider the subgraph $\D_v$ of $\D$ induced by the vertices $\gamma_v$ and $\phi_v$ as well as the union of their in-neighborhoods and out-neighborhoods and add all edges of the transitive closure of $\D_v$ to~$\D$.

\begin{description}
    \item[Leaf] There is one signature with $\omega=0$, empty functions $\Gamma$, $\Phi$, $b_\Gamma$, $b_\Phi$, and empty digraph~$\D$.
    \item[Introduce Node] Let $v$ be the introduced vertex and let $t'$ be the child of $t$. For every valid signature of $t'$, we construct for every possible choice $(\Gamma(v), \Phi(v))$ from $(\R \cup \{\bot\}) \times (\R \setminus \{D\} \cup \{\bot\}) \setminus \{(T,T), (D,Z), (\bot,T)\}$ a signature of $t$ by keeping the weight $\omega$ and adding vertex $\gamma_v$ and, if $\Phi(v) \neq \bot$, vertex $\phi_v$ to $\D$. 
    
    If $\Gamma(v) \in \{D,\bot\}$, then we set $b_\Gamma = \true$ and if $\Phi(v) = \bot$, then we set $b_\Phi = \true$. If $\Phi(v) = Z$, then we add the arc $(\gamma_v, \phi_v)$ to $\D$ and if $\Phi(v) = T$, then we add the arc $(\phi_v,\gamma_v)$ to $\D$.
    \item[Forget Node] Let $v$ be the forgotten vertex and let $t'$ be the child of $t$. For every valid signature of $t'$ where $b_\Gamma(v) = \true$ and $b_\Phi(v) = \true$, we construct a signature of $t$ as follows. We remove $v$ from the functions $\Gamma$, $\Phi$, $b_\Gamma$ and $b_\Phi$.  We construct the new dependency graph by bypassing $v$ in $\D$ and deleting $\gamma_v$ and $\phi_v$ afterwards. If $\Gamma(v) = \bot$, then we increase the weight $\omega$ by one.
    \item[Join Node] Let $t_1$ and $t_2$ be the children of $t$. We say that the signatures $\Omega^1$ of $t_1$ and $\Omega^2$ of $t_2$ with $\Omega^i = (\Gamma^i, \Phi^i, b_\Gamma^i, b_\Phi^i, \allowbreak \D^i, \omega^i)$ are \emph{compatible} if 
    \begin{itemize}
        \item $\Gamma^1 = \Gamma^2$, $\Phi^1 = \Phi^2$,
        \item for every vertex $v$ with $\Gamma(v) \notin \{D,\bot\}$ it holds that $b^1_\Gamma(v) \land b^2_\Gamma(v) = \false$
        \item for every vertex $v$ with $\Phi(v) \neq \bot$ it holds that $b^1_\Phi(v) \land b^2_\Phi(v) = \false$
        \item the union of $\D^1$ and $\D^2$ is acyclic. 
    \end{itemize} For every of those pairs of signatures, we create a signature of $t$ by keeping $\Gamma^1$ and $\Phi^1$, setting $b_\Gamma(v) = b^1_\Gamma(v) \lor b^2_\Gamma(v)$ and $b_\Phi(v) = b^1_\Phi(v) \lor b^2_\Phi(v)$, $\D = \D^1 \cup \D^2$ and $\omega = \omega_1 + \omega_2$.
    \item[Rule Node] Let $t'$ be the child of $t$ and let $v$ be the vertex that is forgotten in the parent of~$t$. For every valid signature $\Omega$ of $t'$, we create signatures of $t$ in the following four-step procedure:
    \begin{enumerate}[(R1)]
        \item If $b_\Gamma(v) = \false$, then create for every $f \in N_t(v)$ with $\Phi(f) = \Gamma(v)$ and $b_\Phi(f) = \false$ a new signature $\Omega(f)$ which is a copy of $\Omega$. Vertex $f$ will be the vertex that colors $v$ blue. If $b_\Gamma(v) = \true$, i.e., the vertex coloring $v$ blue is already fixed, then create only one such copy $\Omega(\varnothing)$. \label{r1}
        \item For every of the signatures $\Omega(f)$ created in Step~(R\ref{r1}), do the following: If $b_\Phi(v) = \false$, then create for every $g \in N_t(v)$ with $\Gamma(g) = \Phi(v)$ and $b_\Gamma(g) = \false$, a new signature $\Omega(f,g)$ which is a copy of $\Omega(f)$. Vertex $g$ will be the vertex that is colored blue by $v$. If $b_\Phi(v) = \true$, i.e., the vertex colored blue by $v$ has already been fixed, then create only one such copy $\Omega(f,\varnothing)$.\label{r2}
        \item For every of the created signatures $\Omega(f,g)$, do the following:\label{r3}
         \begin{enumerate}
        \item If $f \neq \varnothing$, then add $(\phi_f, \gamma_v)$ to $\D$; set $b_\Gamma(v) = b_\Phi(f) = \true$; if $\Phi(v) = T$, then add arc $(\phi_v,\phi_f)$ to $\D$. This means that we apply rule $f \overset{\scriptscriptstyle \Gamma(v)}{\longrightarrow} v$.\label{r3a}
        \item If $g \neq \varnothing$, then add $(\phi_v, \gamma_g)$ to $\D$; set $b_\Gamma(g) = b_\Phi(v) = \true$; if $\Phi(g) = T$, then add arc $(\phi_g,\phi_v)$ to $\D$. This means that we apply rule $v \overset{\scriptscriptstyle \Phi(v)}{\longrightarrow} g$.\label{r3b}
        \item For each $w \in N_t(v)$ with $\Phi(w) \neq \bot$ and $w \neq f$, add arc $(\gamma_v,\phi_w)$ to $\D$.\label{r3c}
        \item For each $w \in N_t(v)$ with $\Gamma(w) = D$, add arc $(\gamma_v,\gamma_w)$ to $\D$.\label{r3c2}
        \item If $\Phi(v) \neq \bot$, then for each $w \in N_t(v)$ with $w \neq g$, add arc $(\gamma_w,\phi_v)$ to $\D$.\label{r3d}
        \item If $\Gamma(v) = D$, then for each $w \in N_t(v)$, add arc $(\gamma_w,\gamma_v)$ to $\D$.\label{r3e}
        \end{enumerate}
    \item  Check for every of the signatures $\Omega(f,g)$ whether its dependency graph $\D$ is acyclic. If not, then set its weight $\omega$ to $\infty$.
    \end{enumerate}
\end{description}

Note that if we construct two signatures of the same node that only differ in the weight~$\omega$, then we only keep the signature with smaller weight.

In the following, we will prove that there is a valid signature of the root node with weight $\omega = k$ if and only if there is an $\R$-forcing set of $G$ of size $k$. First assume that there is a valid signature $\Omega_{\tilde{t}}$ of the root node ${\tilde{t}}$. We define the \emph{signature tree $\T$ of $\Omega_{\tilde{t}}$} to be the set of signatures that contains $\Omega_{\tilde{t}}$ and exactly one signature $\Omega_t$ for every $t \in V(T)$ such that the signature $\Omega_t \in \T$ was created using the signature $\Omega_{t'}$ where $t'$ is a child of $t$. We define $\fR$ to be the set of rules that have been used to compute the signatures of $\T$. In particular, $\fR$ contains the rule $v \d v$ for every vertex $v$ with $\Gamma$-type $D$ in some signature of $\T$ as well as all the rules that have been applied in~(R\ref{r3a}) and (R\ref{r3b}) to compute the signatures of $\T$. Furthermore, we define $\D^*$ to be the union of the dependency graphs of all signatures in $\T$.

The following two lemmas present some properties of $\D^*$.

\begin{lemma}\label{lemma:rule-arcs}
The following statements are true.
    \begin{enumerate}
        \item If $p \z q$ is in $\fR$, then $(\gamma_p, \phi_p), (\phi_p, \gamma_q) \in \D^*$ and for all $r \in N(p)$ with $r \neq q$, it holds that $(\gamma_r, \phi_p) \in \D^*$.
        \item If $p \t q$ is in $\fR$, then $(\phi_p, \gamma_p), (\phi_p, \gamma_q) \in \D^*$ and for all $r \in N(p)$ with $r \neq q$, it holds that $(\gamma_r, \phi_p) \in \D^*$. Furthermore, if there is a rule $s \z p$ in $\fR$, then $(\phi_p, \phi_s) \in \D^*$.
        \item If $p \d p$ is in $\fR$, then for all $r \in N(p)$, it holds that $(\gamma_r, \gamma_p) \in \D^*$.
    \end{enumerate}
\end{lemma}

\begin{proof}
    We only prove the second statement. The other statements can be proved analogously. First observe that $(\phi_p, \gamma_p)$ is inserted into the dependency graph in the introduce node of $p$ since $\Phi(p) = T$.
    
    Next, we show that $(\phi_p,\gamma_q) \in \D^*$. First assume that $p$ is forgotten before $q$. Then, the rule $p \t q$ was applied in the rule node of $p$. Hence, $g$ was chosen to be $q$ in (R\ref{r2}) and we have added $(\phi_p,\gamma_q)$ to $\D$ in (R\ref{r3b}). If $q$ is forgotten before $p$, then the rule $p \t q$ was applied in the rule node of $q$. Hence, $f$ was chosen to be $p$ in (R\ref{r1}) and we have added $(\phi_p,\gamma_q)$ to $\D$ in (R\ref{r3a}).

    Now let $r \in N(p) \setminus \{q\}$. If $r$ is forgotten before $p$, then $p$ is present in the rule node of~$r$. Furthermore, $p$ was not chosen to be $f$ to create the signature for this node in (R\ref{r3a}) since, otherwise, $b_\Phi(p)$ would have been $\false$ beforehand and then set to $\true$ implying that the rule $p \t q$ was neither applied before nor applied afterwards. Hence, we have added $(\gamma_r, \phi_p)$ in (R\ref{r3c}). Analogously, if $p$ is forgotten before $r$, then $r$ was present in the rule node of $p$ and was not chosen to be $g$. Therefore, we have added $(\gamma_r, \phi_p)$ in (R\ref{r3d}).

    Finally assume that there is a rule $s \z p$ in $\fR$. If $s$ is forgotten before $p$, then $p$ was chosen to be $g$ in the rule node of $s$. Hence, we have added $(\phi_p, \phi_s)$ in (R\ref{r3b}). If $p$ is forgotten before $s$, then $s$ was chosen to be $f$ in the rule node of $p$, and $(\phi_p, \phi_s)$ was added in~(R\ref{r3a}). 
\end{proof}

\begin{lemma}\label{lemma:dependency-cycle}
    $\D^*$ is acyclic. 
\end{lemma}

\begin{proof}
    Assume for contradiction that there is a directed cycle in $\D^*$. Let $C$ be such a cycle of minimal length. Let $(x,y)$ be the arc in $C$ that was removed last from some dependency graph $\D_t$ of some signature $\Omega_t$ in $\mathfrak{T}$, i.e., all other arcs of $C$ are also part of $\D_t$ or of the dependency graph $\D_{t'}$ of some descendant of tree node $t$. Since the signature $\Omega_t$ is valid, the graph $\D_t$ is acyclic. Hence, some arc $e = (p,q)$ of $C$ is not present in $\D_t$. The arc $e$ may only vanish from a dependency graph of a tree node $t_1$ to the dependency graph of $t_1$'s parent $t_2$ if $t_2$ is the forget node of a vertex $v$ and $v$ is the vertex that belongs to at least one of the event nodes incident to $e$. Without loss of generality, we may assume that this event node is $q$. We choose $e$ to be an arc that vanishes first, i.e., no other arc of $C$ vanishes in some tree node below $t_1$. Note that some arcs of $C$ might not be present in $\D_{t_1}$ as they are introduced later. However, the arc $(q,s)$ of $C$ has to be present in $\D_{t_1}$. This holds since it cannot be introduced later because the event node $q$ is not present in any dependency graph above $t_1$. Furthermore, the event node $p$ is present in $\D_{t_1}$ since $e$ is part of $\D_{t_1}$. Hence, when bypassing $q$, the arc $(p,s)$ is introduced to the dependency graph $\D_{t_1}$. However, this implies that this arc is also part of $\D^*$, contradicting the fact that $C$ has minimal length.
\end{proof}

These lemmas enable us to prove the existence of an \R-forcing set if we have found a valid signature in the root node.

\begin{lemma}\label{lemma:tw-hin}
    If there is a valid signature $\Omega_{\tilde{t}}$ in the root node ${\tilde{t}}$ with weight $\omega = k$, then there is an $\R$-forcing set of $G$ of size $k$.
\end{lemma}

\begin{proof}
    Consider the signatures in the signature tree $\T$. Note that a vertex of $G$ has the same $\Gamma$-type in all signatures of $\T$. Hence, there are exactly $k$ vertices of $G$ that got the $\Gamma$-type $\bot$. Let $S$ be the set of those vertices. Since forget nodes only adopt those signatures from their child where the forgotten vertex has $b_\Gamma$-value $\true$, every vertex that is not in $S$ has been colored blue by exactly one of the rules in $\fR$. Due to \cref{lemma:dependency-cycle}, the complete dependency graph $\D^*$ is acyclic. Thus, there is a topological sorting of $\D^*$. We claim that there is also a sorting fulfilling certain properties.

    \begin{claim}
        There is a topological sorting $\sigma$ of $\D^*$ such that $\sigma$ starts with the $\gamma$-nodes of the vertices in $S$ and if there is a rule $v \z w$ or $v \t w$ in $\fR$, then $\phi_v$ is the direct predecessor of $\gamma_w$ in the sorting. 
    \end{claim}

    \begin{proof}[Proof of claim]
        There are only two options how a event node $\gamma_x$ can get its first incoming edge. Either it holds that $\Phi(x) = T$ or it holds that $\Gamma(x) \neq \bot$. Since both properties do not hold for the vertices in $S$, their $\gamma$-nodes have in-degree zero and, thus, can be taken at the beginning of the topological sorting.

        Now assume that one of the rules  $v \z w$ or $v \t w$ is in $\fR$. Then, due to \cref{lemma:rule-arcs}, $(\phi_v,\gamma_w) \in \D^*$. We now show that if there is another arc $(x,\gamma_w)$ in $\D^*$ with $x \neq \phi_v$, then there is also an arc $(x,\phi_v)$ in $\D^*$. This implies that we can take $\gamma_w$ directly after $\phi_v$ in $\sigma$.  

        Let $(x,\gamma_w)$ with $x \neq \phi_v$ be the first such arc that is added to some dependency graph. If $(x,\gamma_w)$ is a bypassing arc, then some arcs $(x,y)$ and $(y, \gamma_w)$ have already been present in the dependency graph. Due to the assumption on $(x,\gamma_w)$, it holds that $y = \phi_v$. Therefore, $(x,\phi_v)$ is in $\D^*$. 
        
        So we may assume that $(x,\gamma_w)$ is not a bypassing arc. As $x \neq \phi_v$, there are only two options how such an arc can be created. The first case is that we add arcs $(\gamma_z, \gamma_w)$ in some rule nodes if $\Gamma(w) = D$. However, there is the rule $v \z w$ or $v \t w$ in $\fR$, so $\Gamma(w) \in \ZT$. Thus, this case cannot occur. 
        
        The only remaining case is the following. We add the arc $(\phi_w, \gamma_w)$ in the introduce node of $w$ if $\Phi(w) = T$ and it remains to show that the arc $(\phi_w,\phi_v)$ is in $\D^*$. If $\Phi(w) = T$, then there is a rule $w \t z \in \fR$ for some vertex~$z$. Furthermore, it holds that $\Gamma(w) \in \ZD$ (see possible choices for $(\Gamma(v), \Phi(v))$ in the introduce nodes and \cref{lemma:t-independent}). As we have seen above, the case of $\Gamma(w) = D$ is not possible. So it holds that $\Gamma(w) = Z$. This implies that the rule that colors $w$ blue is the $Z$-rule $v \z w$. Now we can apply \cref{lemma:rule-arcs}. As we have seen, there are the rules $w \t z$ and $v \z w$ in $\fR$. If we choose $p = w$, $q = z$, and $s = v$ in Case~2 of \cref{lemma:rule-arcs}, then the lemma implies that $(\phi_w,\phi_v)$ is in~$\D^*$.
    \end{proof}

    So let $\sigma$ have the form as stated in the claim. We arrange the $Z$- and $T$-rules in $\fR$ according to the order of the respective event nodes in $\sigma$. Note that this is possible since the respective $\phi$- and $\gamma$-nodes appear consecutively in $\sigma$. Furthermore, we place each $D$-rule $v \d v$ at the position of $\gamma_v$. We claim that this ordering of the rules is a valid ordering, i.e., all respective vertices have the right color. This follows from \cref{lemma:rule-arcs}. For example, if we have a rule $p \t q$, then neither $p$ nor $q$ have been colored blue before the rule (due to $(\phi_p,\gamma_p), (\phi_p, \gamma_q) \in \D^*$) and all other neighbors $r$ of $p$ are blue since $(\gamma_r, \phi_p) \in \D^*$. For the other rules, we can use the same arguments. This finalizes the proof.
\end{proof}

It remains to show that an $\R$-forcing set of size $\leq k$ can also be found by our algorithm.

\begin{lemma}\label{lemma:tw-rueck}
    If there is an $\R$-forcing set of $G$ of size $\leq k$, then there is a valid signature $\Omega_{\tilde{t}}$ in the root node ${\tilde{t}}$ with weight $\omega \leq k$.
\end{lemma}

\begin{proof}
    Let $S$ be an $\R$-forcing set of size $\leq k$ and let $\sigma = (R_1, \dots, R_\ell)$ be the sequence of rules of $\R$ applied to $G$. For every tree node, we only consider signatures that match $S$ and $\sigma$, i.e, if $v \in S$, then $\Gamma(v) = \bot$ and if there is a rule $v \x w$ with $X \in \R \setminus \{D\}$, then $\Phi(v) = X$ and $\Gamma(w) = X$ and if there is a rule $v \d v$, then $\Gamma(v) =D$. Furthermore, in every rule node, we apply exactly those rules to the forgotten vertex that are part of $\sigma$. We claim that doing this, we construct a valid signature of the root node $\tilde{t}$ with weight $\omega \leq k$.

    To this end, we define the dependency graph $\D^\sigma$. Here, we add arcs with respect to $\sigma$, i.e., for every rule in $\sigma$, we add the arcs as given in \cref{lemma:rule-arcs}. Furthermore, if rule $p \x q$ is applied before rule $r \y s$, then we add the arcs $(\phi_p, \phi_r)$ and $(\gamma_q,\phi_r)$. Let $\D^\sigma$ be the transitive closure of the digraph containing all these arcs. It is clear that the graph $\D^\sigma$ is acyclic since it represents the linear order of the rules of $\sigma$. Furthermore, the graph $\D^*$ is a subgraph of $\D^\sigma$ and, thus, $\D^*$ is also acyclic. As the dependency graphs of every signature are subgraphs of $\D^*$, they are also acyclic and, thus, all these signatures are valid. Hence, the algorithm has computed a valid signature of the root node $\tilde{t}$. Since we only have used $|S|$ many vertices with $\Gamma$-value $\bot$, the weight of this signature has to be at most $|S| \leq k$.
\end{proof}

Finally, we can state the main result of this section.

\begin{theorem}\label{thm:tw-forcing}
    Let $\R$ be a non-empty subset of $\ZTD$. Then $\FORCE{\R}$ can be solved in $2^{\O(d^2)} \cdot n$ time where $n$ is the number of vertices and $d$ is the treewidth of the given graph.
\end{theorem}

\begin{proof}
    \Cref{lemma:tw-hin,lemma:tw-rueck} show that the algorithm works correctly. So consider the running time. We apply Korhonen's algorithm~\cite{korhonen2021single} to compute a tree decomposition of width $\leq 2d +1$ in $2^{\O(d)} \cdot n$ time. It is easy to see that we can bring such a decomposition into the form using our five node types in $\O(n)$ time. Let $\ell = 2d + 2$. Then for a particular tree node $t$, there are at most $4^\ell \cdot 3^\ell \cdot 2^\ell \cdot 2^\ell \cdot 2^{\O(\ell^2)} \in 2^{\O(d^2)}$ different choices for $(\Gamma, \Phi, b_\Gamma, b_\Phi, \D)$. For every of these choices, we have at most one signature of $t$. Since there are $\O(n)$ tree nodes, the total number of signatures is bounded by $2^{\O(d^2)} \cdot n$. The number of steps used for one of these signatures is polynomial in $d$. Overall, this implies a running time of $\O(d^c) \cdot 2^{\O(d^2)} \cdot n = 2^{\O(d^2)} \cdot n$.
\end{proof}

Note that the algorithm does not only solve the decision problem whether there is an $\R$-forcing set of size $\leq k$ but even computes the smallest size of such a set. Therefore, \cref{thm:lin,thm:local-l} imply the following.

\begin{corollary}\label{corol:grundy-tw}
    \GN{}, \TN{}, \ZN{}, and \LLNC{} can be solved in $2^{\O(d^2)} \cdot n$ time where $n$ is the number of vertices and $d$ is the treewidth of the given graph.
\end{corollary}

At the end of the subsection, we sketch how we can adapt the algorithm to also solve the connected variants of the Grundy domination problems as well as \CON{\FORCE{\R}} and \TOT{\FORCE{\R}}. 

For \TOT{\FORCE{\R}}, we introduce for every vertex of $\Gamma$-type $\bot$ a boolean value that is set to $\true$ if the vertex is adjacent to some other vertex of $\Gamma$-type $\bot$. So whenever we forget such a vertex, we update its boolean value and the boolean values of its neighbors.

For \CON{\FORCE{\R}}, we add some function to the signatures that maps the vertices of $\Gamma$-type $\bot$ to some value in $[d+1]$. This integer represents the connected component of the subgraph induced by the vertices of $\Gamma$-type $\bot$. Whenever we introduce a new vertex of $\Gamma$-type $\bot$, then we join components if this new vertex is adjacent to both of them. Afterwards, we add the vertex to the unique adjacent component or if no such component exists, we open a new one. If a vertex $v$ of $\Gamma$-type $\bot$ is forgotten, we check whether it is the only vertex of its component in the bag. If so, then there a two cases. If $v$ is the only vertex of $\Gamma$-type $\bot$ in the bag, then we forbid to introduce new vertices of $\Gamma$-type $\bot$ in later bags. Otherwise, the $\R$-forcing set is not connected and the signature is made invalid. It is not difficult to see that in both cases the additional effort is bounded by $2^{\O(d^2)}$. Thus, the following holds.

\begin{theorem}\label{thm:conn-tot}
    Let $\R$ be a non-empty subset of $\ZTD$. Then both the problem $\TOT{\FORCE{\R}}$ and the problem $\CON{\FORCE{\R}}$ can be solved in $2^{\O(d^2)} \cdot n$ time where $n$ is the number of vertices and $d$ is the treewidth of the given graph.
\end{theorem}

We use a similar approach for the connected variants of the Grundy domination problems. Here, we keep track of the components of all vertices of $\Gamma$-type $\neq \bot$ since these vertices are part of the dominating sequence, due to the proofs of \cref{thm:lin} given in~\cite{bresar2017grundy,lin2019zero} as well as the proof of \cref{thm:local-l}.

\begin{theorem}
    \CON{\GN}, \CON{\TN}, \CON{\ZN}, and \CON{\LLNC} can be solved in $2^{\O(d^2)} \cdot n$ time where $n$ is the number of vertices and $d$ is the treewidth of the given graph.
\end{theorem}

\subsection{Parameterized Algorithms for Parameter Solution Size}

Bhyravarapu et al.~\cite{bhyravarapu2025parameterized} used their \FPT{} algorithm for \ZFS{} parameterized by the treewidth to give an \FPT{} algorithm for the problem when parameterized by the solution size. To this end, they showed that the pathwidth of a graph with a $\{Z\}$-forcing set of size $k$ is bounded by $k$, a result that has already been given by Aazami~\cite[Theorem 2.2.1]{aazami2008hardness}. Barioli~et~al.~\cite{barioli2013parameters} showed that this even holds for a more general rule set including set $\ZD$.

\begin{lemma}[Barioli et al.~{\cite[Theorem~2.44]{barioli2013parameters}}]\label{lemma:barioli}
    If a graph $G$ has a $\ZD$-forcing set of size $k$, then the pathwidth of $G$ is at most $k$.
\end{lemma}

This allows us to extend the result of Bhyravarapu et al.~\cite{bhyravarapu2025parameterized} for \FORCE{Z} to \FORCE{Z,D}. Note that Bhyravarapu et al.~only claimed a running time of $2^{\O(k^2)} \cdot n^{\O(1)}$.

\begin{theorem}\label{thm:solution-size}
    \FORCE{Z} and \FORCE{Z,D} can be solved in $2^{\O(k^2)} \cdot n$ time.
\end{theorem}

\begin{proof}
    Given an instance $(G,k)$ of one of the two problems, we apply Korhonen's algorithm~\cite{korhonen2021single} to $(G,k)$ which needs  $2^{\O(k)} \cdot n$ time. If it outputs that the treewidth of $G$ is $> k$, then -- due to \cref{lemma:barioli} -- there is neither a $\ZD$-forcing set nor a $\{Z\}$-forcing set of size at most $k$. Otherwise, we apply the algorithm of \cref{thm:tw-forcing} for $d = k$ which needs $2^{\O(k^2)} \cdot n$ time.
\end{proof}

Due to \cref{thm:lin}, the result of \cref{thm:solution-size} implies the following for the duals of \GN{} and \ZN{}.

\begin{corollary}\label{corol:gn-zn}
   \DUAL{\GN} and \DUAL{\ZN} can be solved in $2^{\O(k^2)} \cdot n$ time.
\end{corollary}

Using \cref{thm:conn-tot} instead of \ref{thm:tw-forcing} in the proof of \cref{thm:solution-size} implies the following.

\begin{corollary}\label{corol:solution-size-connected}
    \CON{\FORCE{Z}}, \TOT{\FORCE{Z}}, \CON{\FORCE{Z,D}}, and \TOT{\FORCE{Z,D}} can be solved in $2^{\O(k^2)} \cdot n$ time.
\end{corollary}

Our algorithm can also be used to solve \CONC{\FORCE{D}} in $2^{\O(k^2)} \cdot n$. However, we have seen in \cref{obs:vc} that these problems are equivalent to \CONC{\textsc{Vertex Cover}} which allow much better \FPT{} algorithms (see, e.g., \cite{chen2006improved,cygan2012deterministic}).
 
Unfortunately, we cannot use the idea of \cref{thm:solution-size} to give an \FPT{} algorithm for \FORCE{Z,T}, \FORCE{T,D}, or \FORCE{Z,T,D} since the treewidth of graphs having $\ZT$-forcing sets or $\TD$-forcing sets of size $k$ cannot be bounded by some value $f(k)$. In fact, almost no reasonable graph parameter can be bounded by some $f(k)$ as is shown by the following result.

\begin{proposition}
    Let $\xi$ be an unbounded graph parameter that does not decrease when leaves are added to the graph. Then for every $k \in \N$, there is a graph $G$ with $\xi(G) \geq k$ such that the empty set is a $\ZT$-forcing set and a $\TD$-forcing set of $G$.
\end{proposition}

\begin{proof}
    Since $\xi$ is unbounded, there is a graph $H$ with $\xi(H) \geq k$. We construct $G$ by appending a leaf to every vertex of $H$. Due to the condition on $\xi$, it holds that $\xi(G) \geq \xi(H) \geq k$.

    We start our $\ZT$-forcing procedure and our $\TD$-forcing procedure with a complete white $G$. Now we iteratively apply the $T$-rule to the added leaves. Note that this is possible since every leaf is white and has exactly one white neighbor. After this, all vertices of $H$ are blue. Therefore, we can iteratively apply the $Z$-rule to them as all of them are blue and have exactly one white neighbor. Instead of applying the $Z$-rule to the vertices of $H$, we can also apply the $D$-rule to the leaves.
\end{proof}

Note that the result also holds for total forcing sets as the empty set does not induce isolated vertices. A similar result also holds for connected forcing sets. If the empty set is considered as connected, then this is trivial. Otherwise, coloring one vertex of $H$ blue leads to a connected forcing set.

We are only aware of a few reasonable graph parameters that might decrease when leaves are added to the graph. One example is the minimum degree $\delta(G)$ that has also been considered in the context of parameterized algorithms at least once~\cite{cook20102k2}. It is easy to see that $\delta(G)$ is at most $k+1$ if there is a $\ZTD$-forcing set of size $k$ in $G$. So an \FPT{} algorithm for the parameter $\delta(G)$ would directly imply an \FPT{} algorithm for the solution size. Nevertheless, it is also easy to see that \FORCE{T,D}, \FORCE{Z,T}, and \FORCE{Z,T,D} are \paraNP-hard for $\delta(G)$.\footnote{For a graph $G$, we construct $G'$ by adding two universal vertices $u_1$ and $u_2$ to $G$ and a vertex $t$ that is only adjacent to $u_1$ and $u_2$. Then every forcing set of $G'$ contains at least one vertex of $\{u_1, u_2, t\}$. Conversely, adding $u_1$ to a forcing set of $G$ creates a forcing set of $G'$.}

Although we are not able to adapt the \FPT{} algorithm, we can at least show that the other forcing set problems can be solved in \XP{} time. To this end, we show that we can greedily check whether a given set of vertices is a forcing set.

\begin{lemma}
    Given a graph $G$ and a set $S \subseteq V(G)$, we can decide in polynomial time whether $S$ is a $\ZTD$-forcing set or a $\ZT$-forcing set or a $\TD$-forcing set.
\end{lemma}

\begin{proof}
    Consider $\FORCE{\R}$ where $\R$ is one of the sets $\ZTD$, $\ZT$ or $\TD$. We greedily apply any of the allowed rules if possible. Assume for contradiction that we get stuck although the set $S$ is an $\R$-forcing set. Consider the first applied rule $u \x v$ that was a mistake, i.e., the sequence of applied rules up till this point could be extended to a rule sequence $\sigma$ that colors the whole graph blue, but after applying $u \x v$, this is not possible anymore. Consider the first rule of $\sigma$ that cannot be applied anymore. This must be a rule that needs $v$ to be white. If it is the $D$-rule $v \d v$, then we can remove it from $\sigma$ without any consequences. So we may assume that it is a $T$-rule $v \t w$. Due to \cref{lemma:t-independent}, this implies that $X = Z$ and we can replace the rule $v \t w$ by the rule $v \z w$.
\end{proof}

This theorem implies \XP{} algorithms for all these forcing set problems. Due to the non-monotonicity of \TD-forcing sets mentioned in \cref{sec:lln}, we must check not only sets of size $k$ but also all smaller sets for $\FORCE{T,D}$ and its variants.

\begin{theorem}\label{thm:xp}
    The following problems can be solved in time $n^{\O(k)}$ where $n$ is the number of vertices of the graph and $k$ is the natural number given with the input:
    \begin{itemize}
        \item \FORCE{Z,T,D}, \FORCE{Z,T}, \FORCE{T,D},
        \item \CON{\FORCE{Z,T,D}}, \CON{\FORCE{Z,T}}, \CON{\FORCE{T,D}},
        \item \TOT{\FORCE{Z,T,D}}, \TOT{\FORCE{Z,T}}, \TOT{\FORCE{T,D}},
        \item \DUAL{\TN} and \DUAL{\LLNC}.
    \end{itemize}
\end{theorem}

We have seen in \cref{sec:lln} that we can improve this running time for \FORCE{T,D} to polynomial time if we consider graphs of bounded independence number $\alpha(G)$. Here, we show that all problems mentioned in \cref{thm:xp} can be solved in \FPT{} time when parameterized by $k$ plus the independence number as long as the independence number is given as additional input. A key ingredient is the observation that the pathwidth of a graph is bounded by the sum of the size of its smallest \ZTD-forcing set and its independence number.

\begin{lemma}
     If a graph $G$ has a \ZTD-forcing set of size $k$, then the pathwidth of $G$ is at most $k + \alpha(G)$.
\end{lemma}

\begin{proof}
    Let $S$ be a \ZTD-forcing set of size $k$. Let $A$ be those vertices of $G$ to which a $T$-rule is applied, i.e., vertex $v$ is in $A$ if and only if the forcing process applies some rule $v \t w$ for some $w \in V(G)$. 
    
    Due to \cref{lemma:t-independent}, we know that $|A| \leq \alpha(G)$. Let $B$ be the vertices that are colored blue using $T$-rules. Note that $|B| = |A| \leq \alpha(G)$. It is easy to see that $S \cup B$ is a \ZD-forcing set since the $T$-rules are not needed anymore. Thus, we know that there is a \ZD-forcing set of size at most $k + \alpha(G)$. Due to \cref{lemma:barioli}, the pathwidth of $G$ is at most $k + \alpha(G)$.
\end{proof}

Using the same arguments as for \cref{thm:solution-size,corol:solution-size-connected}, we get the following.

\begin{theorem}\label{thm:fpt-alpha}
    When given the independence number $\alpha(G)$ as additional input, the following problems can be solved in $2^{\O(t^2)} \cdot n$ time where $t = k + \alpha(G)$:
    \begin{itemize}
        \item \FORCE{Z,T,D}, \FORCE{Z,T}, \FORCE{T,D},
        \item \CON{\FORCE{Z,T,D}}, \CON{\FORCE{Z,T}}, \CON{\FORCE{T,D}},
        \item \TOT{\FORCE{Z,T,D}}, \TOT{\FORCE{Z,T}}, \TOT{\FORCE{T,D}},
        \item \DUAL{\TN} and \DUAL{\LLNC}.
    \end{itemize}
\end{theorem}

Note that it might be possible that there is also an \FPT{} algorithm for \FORCE{Z,T} and its variants when parameterized by the solution size, i.e., we might omit the independence number from the parameter. 

\subsection{W[1]-Hardness Results}\label{sec:force-w1}

In this section, we will complement the \XP{} algorithms presented in \cref{thm:xp} by respective lower bounds -- at least for \FORCE{Z,T,D} and \FORCE{T,D}. Furthermore, we show that the problems are \WOne-hard.

\begin{theorem}\label{thm:force-w1}
    The following problems are \WOne-hard on split graphs and on bipartite graphs when they are parameterized by the input value $k$:
    \begin{itemize}
        \item \FORCE{T,D} and \FORCE{Z,T,D},
        \item \CON{\FORCE{T,D}} and \CON{\FORCE{Z,T,D}},
        \item \TOT{\FORCE{T,D}} and \TOT{\FORCE{Z,T,D}},
        \item $\DUAL{\LLNC}$.
    \end{itemize}
    Furthermore, assuming the Exponential Time Hypothesis, there is no $f(k) n^{o(k)}$ time algorithm for these problems for any computable function~$f$ even if the given $n$-vertex graph is a bipartite graph or a split graph.
\end{theorem}

In the following, we will present a proof of that theorem using several lemmas. Again, we reduce from \MCP{}. Let $G$ be the input graph and $V^1, \dots, V^k$ be the partition of $V(G)$ into independent sets. Without loss of generality, we may assume that $k \geq 2$ and that all color classes are of the same size, i.e., we assume that $V^i = \{v^i_1, \dots, v^i_q\}$ for all $i \in [k]$ and some fixed value $q \in \N$. We also may assume that $q \geq 4$ and that all vertices have at least two neighbors in every other color class. This assumption is justified by the way the \WOne-hardness of \MCP{} is proven (see \cite{fellows2009param,pietrzak2003parameterized}). We construct the graph $G'$ as follows:

\begin{figure}
    \centering
    \begin{tikzpicture}
    \scriptsize
    \node[svertex, label={[name=lxi11]-90:$x^i_1(1)$}] (xi11) at (0,0) {};
    \node[svertex, label=-90:$x^i_1(2)$] (xi12) at (1,0) {};

    \node[svertex, label=180:$y^i_1(1)$] (yi11) at (-1.5,1) {};
    \node[svertex, label=180:$y^i_1(2)$] (yi12) at (0,1) {};
    \node[svertex, label=0:$y^i_1(3)$] (yi13) at (1.5,1) {};
    \node[svertex, label=0:$y^i_1(4)$] (yi14) at (3,1) {};

    \begin{scope}[xshift=7cm]
        \node[svertex, label=-90:$x^i_q(1)$] (xiq1) at (0,0) {};
        \node[svertex, label={[name=lxiq2]-90:$x^i_q(2)$}] (xiq2) at (1,0) {};
    
        \node[svertex, label=180:$y^i_q(1)$] (yiq1) at (-1.5,1) {};
        \node[svertex, label=180:$y^i_q(2)$] (yiq2) at (0,1) {};
        \node[svertex, label=0:$y^i_q(3)$] (yiq3) at (1.5,1) {};
        \node[svertex, label=0:$y^i_q(4)$] (yiq4) at (3,1) {};
    \end{scope}

    \node at (4.25,1) {$\cdots$};
    \node at (4.25,-0.26) {$\cdots$};

    \node[svertex, label=90:$s^i(1)$] (si1) at (2,3) {};
    \node[svertex, label=90:$s^i(2)$] (si2) at (3.5,3) {};
    \node[svertex, label=90:$s^i(3)$] (si3) at (5,3) {};
    \node[svertex, label=90:$s^i(4)$] (si4) at (6.5,3) {};

    \draw (xi11) -- (yi11);
    \draw (xi11) -- (yi12);
    \draw (xi11) -- (yi13);
    \draw (xi11) -- (yi14);

    \draw (xi12) -- (yi11);
    \draw (xi12) -- (yi12);
    \draw (xi12) -- (yi13);
    \draw (xi12) -- (yi14);

    \draw (xiq1) -- (yiq1);
    \draw (xiq1) -- (yiq2);
    \draw (xiq1) -- (yiq3);
    \draw (xiq1) -- (yiq4);

    \draw (xiq2) -- (yiq1);
    \draw (xiq2) -- (yiq2);
    \draw (xiq2) -- (yiq3);
    \draw (xiq2) -- (yiq4);

    \draw (si1) -- (yiq1);
    \draw (si2) -- (yiq2);
    \draw (si3) -- (yiq3);
    \draw (si4) -- (yiq4);

    \draw (si1) -- (yi11);
    \draw (si2) -- (yi12);
    \draw (si3) -- (yi13);
    \draw (si4) -- (yi14);
\end{tikzpicture}
    \caption{Selection gadget $\S^i$ of color class $V^i$ in the construction of \cref{thm:force-w1}.
    }
    \label{fig:force-w1}
\end{figure}
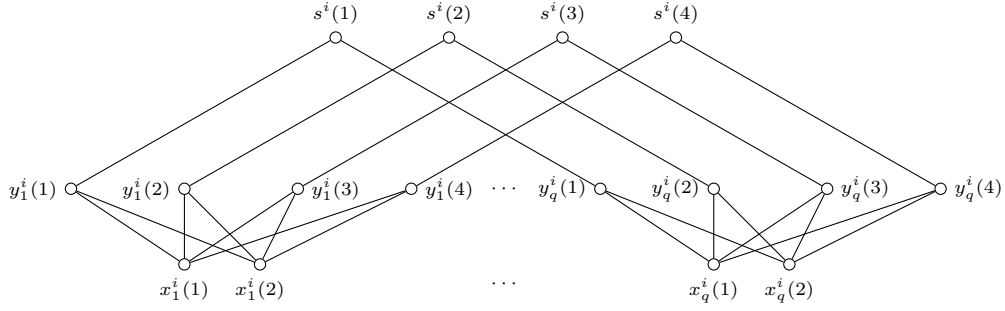

\begin{description}
    \item[Selection Gadgets] For every $i \in [k]$, we have a selection gadget $\S^i$ (see \cref{fig:force-w1} for an illustration). For every $p \in [q]$, it contains two vertices $x^i_p(1)$ and $x^i_p(2)$ -- called \emph{selection vertices} -- as well as four vertices $y^i_p(1), \dots, y^i_p(4)$. For every $p$, these vertices form a complete bipartite graph with the $x$-vertices on the one side and the $y$-vertices on the other side. 
    Additionally, the gadget contains the four vertices $s^i(1), \dots, s^i(4)$. For every $p \in [q]$ and $a \in \{1,2,3,4\}$, the vertex $y^i_p(a)$ is adjacent to the vertex $s^i(a)$. 
    \item[Verification Gadgets] For every two-element set $\{i,j\} \subseteq [k]$, we have two verification gadgets $\C^{ij}(1)$ and $\C^{ij}(2)$. Such a gadget $\C^{ij}(b)$ contains for every edge $v^i_pv^j_r \in E(G)$ an \emph{edge vertex} $w^{ij}_{pr}(b)$. Furthermore, it contains a \emph{verification vertex} $c^{ij}(b)$. The vertex $c^{ij}(b)$ is adjacent to all edge vertices $w^{ij}_{pr}(b)$. Note that we can write $\C^{ij}(b)$ and $\C^{ji}(b)$ interchangeably. Similarly, the vertices $c^{ij}(b)$ and $c^{ji}(b)$ as well as the vertices $w^{ij}_{pr}$ and $w^{ji}_{rp}$ are identical. So as in the proof of \cref{thm:bip}, the indices $i$ and $j$ are bounded to the indices $p$ and $r$ by their order, i.e., the first (second) index in superscript belongs to the first (second) index in subscript.
    \item[Clean-Up Gadget] For every selection vertex $x^i_p(a) \in \S^i$, we have a vertex $\hat{x}^i_p(a)$. We call all these vertices \emph{clean-up vertices}.
    \item[Connectivity Gadget] We have two vertices $z_1$ and $z_2$ as well as two vertices $t_1$ and $t_2$, where $t_1$ and $t_2$ are both adjacent to $z_1$ and $z_2$. We call these four vertices \emph{connectivity vertices}.
\end{description}

Besides the edges within the gadgets, there are also edges between vertices of different gadgets.

\begin{enumerate}[(E1)]
    \item For every $i \in [k]$ and $a \in \{1,2,3,4\}$, vertex $s^i(a)$ is adjacent to all edge vertices and all clean-up vertices.
    \item Every clean-up vertex is adjacent to all verification vertices.
    \item The clean-up vertex $\hat{x}^i_p(a)$ is adjacent to the selection vertex $x^i_p(a)$.
    \item Selection vertex $x^i_p(a)$ is adjacent to edge vertex $w^{ij}_{rs}(b)$ if and only if $r = p$.\label{item:force-w1-edge1}
    \item The vertices $z_1$ and $z_2$ are adjacent to all vertices $s^i(a)$, to all selection vertices and to all verification vertices.\label{e5}
\end{enumerate}

Observe that $G'$ is bipartite since we can partition $V(G')$ into two independent sets $A$ and $B$. Set $A$ contains all selection vertices, all verification vertices, all vertices $s^i(a)$ as well as the vertices $t_1$ and $t_2$. Set $B$ contains all $y$-vertices, all edge vertices, all clean-up vertices as well as the vertices $z_1$ and $z_2$. We also consider the graph $G''$ which is constructed from $G'$ by making the set $(A \setminus \{t_1,t_2\}) \cup \{z_1,z_2\}$ a clique. Note that $G''$ is a split graph. A summary of the adjacencies in $G'$ and $G''$ is given in \cref{tab:force-w1}.

\newcommand{\all}{\ding{51}}
\newcommand{\no}{\ding{55}}

\newcommand{\checkedcross}{%
  \mathchoice%
    {\ooalign{\all\cr\hidewidth\scalebox{0.6}{\no}\hidewidth}}%
    {\ooalign{\all\cr\hidewidth\scalebox{0.6}{\no}\hidewidth}}%
    {\ooalign{\scalebox{0.7}{\all}\cr\hidewidth\scalebox{0.45}{\no}\hidewidth}}%
    {\ooalign{\scalebox{0.5}{\all}\cr\hidewidth\scalebox{0.35}{\no}\hidewidth}}%
}
\newcommand{\nosplit}{\textsubscript{\no}\textsuperscript{\all}}

\begin{table}[]
    \centering
    \caption{Overview of the adjacencies within the graphs $G'$ and $G''$ constructed in the proof of \cref{thm:force-w1}. For every row, the entries show to which vertices of the respective class the vertex given at the front of the row is adjacent. The symbol~\all{} represents that all edges are present both in $G'$ and in $G''$. The symbol~\no{} represents that neither in $G'$ nor in $G''$ there is an edge to that class of vertices. The symbol~{\nosplit{} represents that there are no edges in $G'$ but all edges in $G''$. In all other cases, the adjacent vertices are given explicitly. Here, all parameters that are not used for the vertex at the front of the row can be chosen arbitrarily.}}
    \begin{tabular}{c  c   c   c   c   c   c  c  c  c}
    \addlinespace
		\toprule
         & $s$  & $x$ & $y$ & $w$ & $c$ & $\hat{x}$ & $ z $ & $t$  \\[0.5ex]\toprule
$s^i(a)$ & \nosplit & \nosplit  & all $y^i_p(a)$ & \all & \nosplit & \all & \all & \no  \\[0.5ex] \midrule 
$x^i_p(a)$ & \nosplit & \nosplit & all $y^i_p(b)$ & all $w^{ij}_{pr}(b)$ & \nosplit & $\hat{x}^i_p(a)$ & \all & \no \\[0.5ex] \midrule
$y^i_p(a)$ & $s^i(a)$ & all $x^i_p(b)$ & \no & \no & \no & \no & \no & \no \\[0.5ex] \midrule
$w^{ij}_{pr}(a)$ & \all & all $x^i_p(b)$, all $x^j_r(b)$ & \no & \no & $c^{ij}(a)$ & \no & \no & \no \\[0.5ex] \midrule
$c^{ij}(a)$ & \nosplit & \nosplit & \no & all $w^{ij}_{pr}(a)$ & \nosplit & \all & \all & \no \\[0.5ex] \midrule
$\hat{x}^i_p(a)$ & \all & $x^i_p(a)$ & \no & \no & \all & \no & \no & \no \\[0.5ex] \midrule
$z_i$ & \all & \all & \no & \no & \all & \no & \nosplit & \all \\[0.5ex] \midrule 
$t_i$ & \no & \no & \no & \no & \no & \no & \all & \no \\[0.5ex] \bottomrule
    \end{tabular}
    \label{tab:force-w1}
\end{table}

In the following, we will prove that there is a multicolored clique in $G$ of size $k$ if and only if there is a respective forcing set of $G'$ and of $G''$ of size $2k + 1$.

Similar as in the proof of \cref{thm:bip}, we will distinguish different phases of the coloring process. In the \emph{selection phase}, we choose some vertices of the selection gadgets to be part of the forcing set and these vertices should represent a multicolored clique. Then, in the \emph{verification phase} all verification vertices should be colored blue which we will show is only possible if our choice in the selection phase has been correct. Finally, we will have a \emph{clean-up phase}, which colors all remaining vertices blue.

We start by showing that a multicolored clique implies a connected $\TD$-forcing set of $G'$ and $G''$ of size $2k + 1$. Note that such a set is also a total $\TD$-forcing set as well as a connected and total $\ZTD$-forcing set.

\begin{lemma}\label{lemma:force-w1-direction1}
    If there is a multicolored clique $\K = \{v^1_{p_1}, \dots, v^k_{p_k}\}$ in $G$, then there is a connected $\TD$-forcing set of $G'$ and $G''$ of size $2k + 1$.
\end{lemma}

\begin{proof}
    We claim that the set $\{x^1_{p_1}(1), x^1_{p_1}(2), \dots, x^k_{p_k}(1), x^k_{p_k}(2)\} \cup \{z_1\}$ is a connected $\TD$-forcing set of $G'$ and $G''$. First observe that the set induces a connected subgraph since all $x$-vertices are adjacent to $z_1$ in both graphs, due to (E\ref{e5}).
    
    We start our $\TD$-forcing procedure by coloring all these vertices blue. Then we color the connectivity vertices blue. Observe that $t_1$ is white and has only one white neighbor $z_2$, so we can apply rule $t_1 \t z_2$. Now both $t_1$ and $t_2$ have only blue neighbors, so we can apply the rules $t_1 \d t_1$ and $t_2 \d t_2$.
    
    For every $i \in [k]$, the vertices $y^i_{p_i}(a)$ are white and have exactly one white neighbor: vertex $s^i(a)$. Thus, we can apply all the rules $y^i_{p_i}(a) \t s^i(a)$. Afterwards, all vertices $s^i(a)$ are blue. 

    Since $\K$ is a clique, there are the edge vertices $w^{ij}_{p_ip_j}(1)$ and $w^{ij}_{p_ip_j}(2)$ for every $\{i,j\} \subseteq [k]$. These vertices are white and have exactly one white neighbor: $c^{ij}(1)$ and $c^{ij}(2)$, respectively. Thus, we can apply the rules $w^{ij}_{p_ip_j}(1) \t c^{ij}(1)$ and $w^{ij}_{p_ip_j}(2) \t c^{ij}(2)$. Afterwards, all verification vertices are blue. Therefore, every clean-up vertex $\hat{x}^i_r(a)$ with $r \neq p_i$ is white and has exactly one white neighbor: vertex $x^i_r(a)$. Hence, we can apply the rules $\hat{x}^i_r(a) \t x^i_r(a)$. The remaining white vertices are the $y$-vertices, the edge vertices, and the clean-up vertices. Observe that all these vertices have only blue neighbors. Thus, we can make them all blue by applying the $D$-rule.
\end{proof}

For the reverse direction, it is sufficient to prove that a \ZTD-forcing set of $G'$ and $G''$ of size $\leq 2k + 1$ implies a multicolored clique of size $k$ in $G$ since every \TD-forcing set is also a \ZTD-forcing set. 

So we assume in the following that $S$ is a (not necessarily connected or total) \ZTD-forcing set of $G'$ or $G''$ of size $\leq 2k + 1$. 

\begin{lemma}\label{lemma:force-w1:conn}
    The set $S$ contains at least one connectivity vertex.
\end{lemma}

\begin{proof}
    Assume for contradiction that this is not the case. Let $v$ be the first connectivity vertex that becomes blue. Observe that this cannot happen with the $D$-rule $v \d v$ since $v$ is adjacent to at least two other connectivity vertices that are still white. So there is a vertex $w$ such that either rule $w \t v$ or $w \z v$ is applied to color $v$. However, if $v$ is one of the $z$-vertices, then $w$ is also adjacent to the other $z$-vertex. If $v$ is one of the $t$-vertices, then $w$ is also adjacent to the other $t$-vertex. So in both cases, $w$ had two white neighbors when the rule is applied; a contradiction.
\end{proof}

Next, we show that all other vertices of $S$ have to be part of some selection gadget.

\begin{lemma}\label{lemma:force-w1:sel1}
     For every $i \in [k]$, $S$ contains exactly two vertices of selection gadget $\S^i$. Therefore, $S$ consists of exactly one connectivity vertex and $2k$ vertices from selection gadgets.
\end{lemma}

\begin{proof}
    Assume for contradiction that there is a selection gadget $\S^i$ that contains less than two vertices of~$S$. Since we have assumed that $q \geq 4$, every vertex of $\S^i$ has at least three neighbors in $\S^i$. Hence, every vertex has at least two white neighbors at the beginning and we cannot color any of the vertices of $\S^i$ blue by applying one of the rules to a vertex of $\S^i$. Thus, the first vertex of $\S^i$ that is colored blue by some rule has to be colored from outside of $\S^i$. However, all neighbors of vertices of $\S^i$ that are not part of $\S^i$ are also neighbors of the four vertices $s^i(1), \dots, s^i(4)$ and at least three of them are still white. So none of the vertices outside of $\S^i$ can color some vertex in $\S^i$ blue. This contradicts the fact that $S$ is a \ZTD-forcing set.

    Summarizing, every selection gadget contains at least two vertices of $S$. Due to \cref{lemma:force-w1:conn}, at most $2k$ vertices of $S$ can be contained in selection gadgets. Since there are $k$~selection gadgets, every selection gadget contains exactly two vertices of $S$.
\end{proof}

Next, we show that the two $S$-vertices in $\S^i$ represent exactly one vertex of~$G$.

\begin{lemma}\label{lemma:force-w1:sel2}
    For every $i \in [k]$, there exists a unique $p_i \in [q]$ such that $x^i_{p_i}(a) \in S$ for some $a \in \{1,2\}$.
\end{lemma}

\begin{proof}
    Let $v$ be the first vertex that colors some vertex of $\S^i \setminus S$ blue. Due to \cref{lemma:force-w1:sel1}, there are exactly two vertices of $S$ in $\S^i$. So at least two vertices of $s^i(1), \dots, s^i(4)$ are not in~$S$. Similar as we have observed in the proof of \cref{lemma:force-w1:sel1}, no neighbor of $\S^i$ outside of $\S^i$ can be $v$ since all of them have at least two white neighbors in $\S^i$. Therefore, $v$ is in $\S^i$. As there are only two vertices of $S$ in $\S^i$, $v$ can have at most three neighbors in $\S^i$. Thus, $v$ is some vertex~$y^i_{r}(a)$. As two of its neighbors must be in $S$, at least one of them is a vertex $x^i_{r}(a')$ and the other is different from any $x^i_{r'}(a'')$ with $r' \neq r$.
\end{proof}

We now consider the verification gadgets. We first observe that edge vertices cannot become blue before their respective verification vertex becomes blue.

\begin{lemma}\label{lemma:force-w1:ver1}
    If an edge vertex in $\C^{ij}(b)$ becomes blue, then $c^{ij}(b)$ is already blue.
\end{lemma}

\begin{proof}
    Assume for contradiction that $w^{ij}_{pr}(b)$ is the first vertex in $\C^{ij}(b)$ that becomes blue. Due to \cref{lemma:force-w1:sel1}, $w^{ij}_{pr}(b)$ is not in $S$. Hence, it is made blue by applying some of the rules in $\ZTD$. If this rule is a $D$-rule, then all neighbors of $w^{ij}_{pr}(b)$ are already blue. In particular, this holds for $c^{ij}(b)$; a contradiction to our assumption.
    
    Otherwise, there must be some neighbor of $w^{ij}_{pr}(b)$ that has only one white neighbor. However, we assumed above that every vertex of $G$ has at least two neighbors in every other color class. Therefore, all neighbors of $w^{ij}_{pr}(b)$ have at least one further neighbor among the edge vertices of $\C^{ij}(b)$. Since we have chosen $w^{ij}_{pr}(b)$ to be the first edge vertex to become blue in $\C^{ij}(b)$, the other vertex is also white and we neither can apply the $Z$-rule nor the $T$-rule to color $w^{ij}_{pr}(b)$ blue. This contradicts our assumption.
\end{proof}

Now, we show that the verification phase -- i.e., the coloring of the verification vertices -- has to happen before any further selection vertex is colored blue.

\begin{lemma}\label{lemma:force-w1:ver2}
    Let $i \in [k]$ and let $p_i \in [q]$ be an index such that $x^i_{p_i}(1)$ or $x^i_{p_i}(2)$ is in $S$. If one vertex $x^i_{r}(a)$ with $r \neq p_i$ becomes blue, then for all $j \in [k]$ with $j \neq i$ and all $b \in \{1,2\}$ it holds that $c^{ij}(b)$ is already blue. 
\end{lemma}

\begin{proof}
    Let $x^i_{r}(a)$ be the first vertex with $r \neq p_i$ that becomes blue  and let $x^i_{r}(a')$ be the unique other vertex in $\S^i$ with $a' \neq a$. Due to \cref{lemma:force-w1:sel2}, neither $x^i_{r}(a)$ nor $x^i_{r}(a')$ are in $S$. Thus, $x^i_{r}(a)$ becomes blue through some rule from $\ZTD$ and $x^i_{r}(a')$ is still white when this happens. First assume that $x^i_{r}(a)$ is colored using a $D$-rule. As vertex $v^i_r$ has some neighbor in every other color class of $G$, this would imply that every verification gadget $\C^{ij}(b)$ has at least one edge vertex that is already blue. Due to \cref{lemma:force-w1:ver1}, all the verification vertices $c^{ij}(b)$ are also already blue.
    
    So we may assume that $x^i_r(a)$ is not colored by some $D$-rule. Then there is a neighbor of $x^i_{r}(a)$ whose only white neighbor is $x^i_{r}(a)$. In particular, this implies that the neighbor is neither adjacent to $x^i_{r}(a')$ nor to any other white selection vertex of $\S^i$. Note that there are at least four other white selection vertices in $\S^i$, due to \cref{lemma:force-w1:sel2} and the fact that $q \geq 4$. Both in $G'$ and $G''$, there is only one neighbor of $x^i_{r}(a)$ that fulfills this condition: the clean-up vertex~$\hat{x}^i_{r}(a)$. However, this vertex is adjacent to all verification vertices and, thus, they all must be blue already.
\end{proof}

Finally, this allows us to prove the existence of a multicolored clique in $G$.

\begin{lemma}\label{lemma:force-w1-direction2}
    Let $p_1, \dots, p_k$ be values in $[q]$ such that for each $i \in [k]$ there is some vertex $x^i_{p_i}(a)$ in $S$. Then the set $\{v^1_{p_1}, \dots, v^k_{p_k}\}$ induces a clique in $G$.
\end{lemma}

\begin{proof}
    Let $\{i,j\} \subseteq [k]$ with $i \neq j$ be arbitrary. It is sufficient to show that $v^i_{p_i}v^j_{p_j} \in E(G)$. Let $c^{ij}(b)$ be the first of the two vertices $c^{ij}(1)$ and $c^{ij}(2)$ that is colored blue and let $c^{ij}(b')$ be the other. Again, we can conclude using \cref{lemma:force-w1:sel1} that $c^{ij}(b)$ has been colored blue by applying some rule of $\ZTD$. 
    Due to \cref{lemma:force-w1:ver1}, all edge vertices of $\C^{ij}(b)$ are white when $c^{ij}(b)$ becomes blue. Thus, $c^{ij}(b)$ was not colored blue by a $D$-rule. Since $c^{ij}(b')$ is still white when $c^{ij}(b)$ becomes blue, the vertex that colors $c^{ij}(b)$ blue cannot be adjacent to $c^{ij}(b')$. Thus, it is an edge vertex $w^{ij}_{rt}(b)$. This implies that the vertices $x^i_r(1)$, $x^i_r(2)$, $x^j_t(1)$, and $x^j_t(2)$ are already blue. Due to \cref{lemma:force-w1:ver2}, it holds that $r = p_i$ and $t = p_j$ since otherwise $c^{ij}(b)$ would have already been blue. Hence, the edge $v^i_{p_i}v^j_{p_j}$ exists in $G$.
\end{proof}

The combination of \cref{lemma:force-w1-direction1} with \cref{lemma:force-w1:sel2,lemma:force-w1-direction2} and the fact that $G'$ and $G''$ can be constructed in polynomial time shows that the given construction is a proper \FPT{} reduction from \MCP{} to all the problems mentioned in \cref{thm:force-w1}. This concludes the \WOne-hardness proof. Note that the reduction increases the parameter only linearly. Thus, we can apply \cref{thm:lower} to get the lower bound assuming the Exponential Time Hypothesis.

\bigskip
An adaption of this proof to \FORCE{Z,T} seems to be non-trivial. The main challenge is how we can substitute the $D$-rules that have been used in \cref{lemma:force-w1-direction1} to color the edge vertices, $y$-vertices and clean-up vertices blue.

Since our reduction is a polynomial-time reduction, we can also conclude \NP-completeness of the problems on bipartite graphs and split graphs.

\begin{corollary}
    The following problems are \NP-complete on bipartite graphs and on split graphs:
     \begin{itemize}
        \item \FORCE{T,D} and \FORCE{Z,T,D},
        \item \CON{\FORCE{T,D}} and \CON{\FORCE{Z,T,D}},
        \item \TOT{\FORCE{T,D}} and \TOT{\FORCE{Z,T,D}},
        \item \LLN{} and \LN.
    \end{itemize}
\end{corollary}

Note that the \NP-completeness of \FORCE{Z,T,D} and \LN{} on these graph classes has already been shown earlier~\cite{bresar2017grundy,bresar2020grundy}.

\section{Conclusion}

We have presented an extensive study on the parameterized complexity of Grundy domination problems and their parametric duals called zero forcing problems (see \cref{tab:results} in the introduction for an overview of the results). Nevertheless, some open questions remain. Probably the most intriguing of them is the parameterized complexity of \FORCE{Z,T}, or equivalently, of \DUAL{\TN{}}. It remains open whether the \XP{} algorithm given in \cref{thm:xp} can be complemented with \WOne-hardness or whether we can also give an \FPT{} algorithm.

Further questions concern lower bounds for running times. It remains open whether the \FPT{} algorithms presented in \cref{sec:zero} whose running time exponents are all quadratic in the parameter can be improved to linear exponents. Furthermore, it remains open whether there are polynomial kernels for these parameterizations. Similar questions occur for the \WOne-hard problems. For \LLN{} as well as \DUAL{\LLNC}, we have shown that there is no $n^{o(k)}$ algorithm assuming the Exponential Time Hypothesis, while they can be solved in $n^{\O(k)}$ time. Similar results can be obtained for \PDS{}~\cite{guo2008improved,kneis2006parameterized}\footnote{The result is not stated explicitly in~\cite{guo2008improved,kneis2006parameterized}. However, both papers present parameter-preserving reductions from \textsc{Dominating Set}, and \cref{thm:lower} also holds for that problem~\cite{cygan2015param,lokshtanov2011lower}.} and \UD{} \cite{araujo2023parameterized,dublois2022upper}. However, our \WOne-hardness proofs for the other Grundy domination problems only imply that they cannot be solved in $n^{o(\sqrt[3]{k})}$ time, while all these problems can straightforwardly be solved in $n^{\O(k)}$ time. To best of our knowledge, there is also no tight lower bound known for \UI{}. The used reductions only imply that it cannot be solved in $n^{o(\sqrt{k})}$ time~\cite{downey2000complexity,jiang2012parameterized}.

There has been quite some research on the classical complexity of Grundy domination problems on certain graph classes. This led to \NP-hardness results on several classes (see \cref{fig:reductions} for an overview). For some of these cases, we were able to show that either the respective Grundy domination problem or its parametric dual is also \WOne-hard. For \TN{} and \LN{} on co-bipartite graphs, we have seen that only the primal parameterization is \WOne-hard while the dual parameterization is \FPT{}. It remains open whether there is some non-trivial class where both the primal and the dual parameterization is \WOne-hard or both are in \FPT{}.

Further questions concern the adaption of our results to other variants of dominating sequences and forcing sets. For example, Haynes and Hedetniemi~\cite{haynes2021vertex} introduced \emph{double dominating sequences}. Algorithmic results on these sequences have been given in~\cite{bresar2022computational,sharma2025double,torres2026grundy}. Herrman and Smith generalized this notion to \emph{$k$-dominating sequences}~\cite{herrman2022extending}, which are related to \emph{$k$-forcing sets} introduced by Amos et al.~\cite{amos2015upper}.

\bibliography{lit}

\end{document}